\documentclass{article} 
\usepackage{iclr2027_conference}
\usepackage{times}

\usepackage{amsmath,amsfonts,bm}
\usepackage{amsthm}
\usepackage{mathtools}
\usepackage{thmtools}

\def\eqref#1{equation~\ref{#1}}

\def\1{\bm{1}}

\def\mB{{\bm{B}}}

\def\mI{{\bm{I}}}

\DeclareMathAlphabet{\mathsfit}{\encodingdefault}{\sfdefault}{m}{sl}
\SetMathAlphabet{\mathsfit}{bold}{\encodingdefault}{\sfdefault}{bx}{n}

\newcommand{\E}{\mathbb{E}}

\newcommand{\R}{\mathbb{R}}

\newcommand{\bP}[2][]{\Pr\ifthenelse{\isempty{#1}}{}{_{#1}}\left[#2\right]}
\newcommand{\bE}[2][]{\mathop\mathbb{E}\ifthenelse{\isempty{#1}}{}{_{#1}}\left[#2\right]}
\newcommand{\bI}[2][] {\mathop\mathbb{I}\ifthenelse{\isempty{#1}}{}{_{#1}}\left[#2\right]}

\newtheorem{theorem}{Theorem}[section]
\newtheorem{lemma}[theorem]{Lemma}
\newtheorem{fact}[theorem]{Fact}

\newtheorem{corollary}[theorem]{Corollary}
\newtheorem{remark}[theorem]{Remark}

\DeclarePairedDelimiter{\inner}{\langle}{\rangle}
\DeclarePairedDelimiter{\norm}{\lVert}{\rVert}
\DeclarePairedDelimiter\abs{\lvert}{\rvert}
\iclrfinalcopy

\usepackage{hyperref}
\usepackage{cleveref}
\usepackage{url}
\usepackage{natbib}
\usepackage{amssymb}
\usepackage{wrapfig}

\usepackage{color}
\usepackage{bbm}

\usepackage{lipsum}
\usepackage{enumitem}
\usepackage[ruled,vlined]{algorithm2e}

\allowdisplaybreaks

\title{Finite-Precision Gram-Schmidt Walks}

 \author{Emile Anand, Jan van den Brand, Peter Chen \\
 Georgia Institute of Technology\\
\texttt{\{emile,vdbrand,pchen423\}@gatech.edu}
}

\newcommand{\clamp}{\epsilon_{t}}
\newcommand{\ec}{\epsilon_{iter}}

\newcommand{\els}{\epsilon_{LS}}
 
 \date{}
\begin{document}

\maketitle
        
\newcommand{\thinking}[1]{\\\noindent\textit{\underline{thinking}: #1}}
\renewcommand{\epsilon}{\varepsilon}

\begin{abstract}
    The Gram-Schmidt Walk is a randomized vector-balancing algorithm whose subgaussian guarantees support applications in discrepancy, experimental design, and data compression; however, these theoretical guarantees are established in exact arithmetic, whereas implementations must approximate least-squares directions, boundary updates, and sampling probabilities in finite precision. This is important because small numerical errors can change which coordinates freeze and thereby alter the subsequent trajectory. We analyze the concentration of the perturbed Gram-Schmidt walk directly under bounded, potentially biased and history-dependent errors. For $n$ input vectors of Euclidean norm at most one, we obtain a modified MGF bound depending on key error sources which recovers the original result as the error goes to zero. We also construct a full-column-rank instance in which bounded update errors produce bias of order $\min\{n^2\varepsilon_{\mathrm{iter}},n\}$, showing that updates accumulate error unavoidably under this model. Finally, we validate our findings in a variety of settings by ablating on the bit precision and problem size.\looseness=-1
\end{abstract}

\section{Introduction}
Selecting a small representative subset of a dataset is an important task in algorithmic statistics \citep{xiao2026unifyingdatasetpruningdistillation,göttlich2024usinglowdiscrepancypointsdata,chazelle2000discrepancy}, which requires preserving many statistics with relatively few points. Discrepancy theory offers a useful perspective on this task: assign signs to data vectors so that their signed sum is small. Given some matrix $A\in\R^{m\times n}$, we seek to assign weights $\{1,-1\}$ to its columns such that the sum is minimized in the infinity norm. In other words, we ask how to split the columns into two groups such that both groups have roughly the same sum as measured by the infinity norm of the difference. We call this quantity the discrepancy of the matrix,  
\begin{equation}
    \operatorname{disc}(A) = \min_{x \in \{-1,1\}^n} \norm{A x}_\infty
\end{equation} 
A balanced partition can serve as a compression step if its cardinality or weights are controlled, and finding such a coloring has natural applications in machine learning such as coresets for large datasets \citep{karnin2019discrepancycoresetssketchesmachine, liberty2026nearlyoptimalattentioncoresets,ren2026streaming}, attention approximation \citep{kochetkova2026streamingattentionapproximationdiscrepancy,anand2026smatattentionstructuredlongcontextsequence}, covariate balance in experimental design \citep{harshaw2023balancingcovariatesrandomizedexperiments}, and distribution thinning \citep{carrell2026lowrankthinning}, where one wishes to summarize a large dataset with a few carefully chosen points. In theoretical computer science, this problem has found applications to spectral sparsification \citep{lau2025spectral}, approximation algorithms \citep{eisenbrand2012binpackingdiscrepancypermutations}, computational geometry \citep{Chazelle2004TheDM}, Monte Carlo integration \citep{dick2014discrepancy}, and matrix-vector multiplication \citep{lingas2024booleanmatrixmultiplicationhighly,anand2026structuralcomplexitymatrixvectormultiplication}. Intuitively, if we can split a dataset into two balanced parts where the whole dataset can be summarized by one of the halves, then we can recursively repeat this procedure to find a small approximation of the dataset.\looseness=-1

One such algorithm is the Gram-Schmidt Walk of \cite{bansal2017gramschmidtwalkcurebanaszczyk} which recovers a seminal result of \cite{294762.294765} by allowing one to sample $x$ from a certain distribution such that $A x$ concentrates tightly around the zero vector. Given $B=[b_1,\ldots,b_n]\in\mathbb{R}^{m\times n}$ with $\|b_i\|_2\leq 1$ and an initial fractional assignment $z_0\in[-1,1]^n$, the walk returns $z\in\{1,-1\}^n$. At a high level, the algorithm takes a random walk in the unit cube $[-1, 1]^n$ with entries becoming integral as the algorithm progresses. Once an entry becomes integral we say that it is frozen, with future steps not moving along this dimension. The steps of the random walk are chosen to be mean zero while not increasing discrepancy too much among the alive entries such that the final output has small discrepancy. \cite{harshaw2023balancingcovariatesrandomizedexperiments} showed that for every $v\in\mathbb{R}^m$,
\begin{equation}
    \mathbb{E}\left[\exp\left(\langle B(z-z_0),v\rangle\right)\right]
  \leq \exp\left(\frac12\|v\|_2^2\right).
\end{equation}

Thus every fixed linear projection of the discrepancy has Gaussian-type tails with a dimension-independent scale. This property is also useful as a component of more elaborate algorithms. For instance, \citet{carrell2026lowrankthinning} develop Gram-Schmidt-based thinning procedures and implementations that accelerate repeated balancing steps.

The concentration guarantees above are established for an idealized Gram-Schmidt Walk operating in exact arithmetic: specifically, the least-squares directions and boundary steps are computed exactly, and each random choice follows its prescribed probability, and these guarantees do not automatically extend to the finite precision process which would be executed on a computer that uses a small number of bits. In particular, these aforementioned theoretical guarantees do not quantify how much numerical error the walk can tolerate due to the sensitivity of its movement within the hypercube, or how much working precision is sufficient to preserve the desired discrepancy bound. This motivates a finite-precision analysis that explains how numerical errors affect concentration and provides sufficient conditions for preserving a desired discrepancy guarantee.

Indeed, implementing these steps in practice introduces errors due to finite-precision. These errors can manifest in the computed direction, the boundary update, and the probabilities of hitting each boundary of the unit hypercube, and can affect the current iterate; however, more seriously, they can also affect which coordinate freezes, which in turn changes the least-squares system (which is also solved by maintaining a Cholesky factorization for improved runtimes) that must later be solved. Therefore, in the worst-case, the perturbed trajectory can deviate significantly from the ideal trajectory.
 It is not clear a priori that subgaussianity should be preserved under approximate arithmetic, as small perturbations in position can potentially change which entries become frozen in a given iteration, changing the path that the walk takes over the rest of the algorithm. So, any analysis that follows the perturbed trajectory alongside one exact trajectory may lose the structure responsible for concentration. Moreover, deterministic rounding may not preserve conditional mean zero, so an implementation can accumulate systematic discrepancy even when the individual update errors are small. 
Numerical sensitivity is already familiar from classical Gram-Schmidt orthogonalization, which can suffer severe loss of orthogonality in finite precision, especially for nearly linearly dependent inputs \citep{giraud2005rounding}. The Gram-Schmidt walk introduces a further challenge: numerical errors can affect both the computed directions and the branching and freezing decisions, changing the subsequent trajectory and the distribution of the final coloring. Understanding its numerical robustness therefore requires connecting local computational errors to the concentration of the final discrepancy.\looseness=-1

These observations expose a gap between theory and practice for the Gram-Schmidt Walk. The classical concentration guarantees assume exact arithmetic, while practical implementations use approximate least-squares directions, rounded updates, and inexact sampling probabilities. These errors can change which coordinates freeze and hence which subproblems the algorithm solves next, so the exact-arithmetic guarantees need not hold for the resulting finite-precision output distribution.

 We address this gap by studying a model with  updates $z_{t+1}\!=\!z_t\!+\!\hat\delta_t\hat u_t\!+\!\epsilon_t$ for $\norm{\epsilon_t}_\infty\!\leq\!\varepsilon_{\mathrm{iter}}$ where $\hat u_t$ and $\hat\delta_t$ are computed quantities, and $\epsilon_t$ includes rounding, clamping, and freezing corrections. Uniform budgets $\els$ and $\varepsilon_{\mathrm{pr}}$ additionally control least-squares image error and conditional branch-probability error, respectively. Our analysis permits biased, history-dependent errors and changes to the freezing trajectory. Section~\ref{sec finite precision gs walk} specifies the complete model, including feasibility and freezing requirements. Our contributions address the theory-practice gap in two directions:

\begin{enumerate}[leftmargin=*, itemsep=2pt, topsep=2pt, parsep=0pt, label=(\arabic*)]
    \item \textbf{Theoretical guarantees.} We establish concentration guarantees for the Gram-Schmidt Walk under biased, history-dependent perturbations. The key technical challenge is that numerical errors can change freezing decisions and subsequent least-squares directions, while destroying the exact conditional unbiasedness used in prior analyses. We address these interacting effects through a phase-based analysis that controls accumulated conditional drift without requiring proximity to an ideal execution. While our bounds retain the sharp exact-arithmetic quadratic term in the log moment-generating function, we further reveal how the discrepancy quality depends on arithmetic precision, the bit complexity, and the error of the least-squares solver.
    \item \textbf{Empirical validation.} Our experiments illustrate the precision-dependent discrepancy in practical numerical implementations. We show that finite precision can \emph{significantly} degrade the discrepancy achieved by the Gram-Schmidt Walk, and characterize the sensitivity of this degradation on the  precision and matrix conditioning. By measuring empirical discrepancy tails, we evaluate when approximate arithmetic preserves output quality and when it is a limitation.
\end{enumerate}

\subsection{Related Works}

\textbf{Gram-Schmidt walk.} \cite{bansal2017gramschmidtwalkcurebanaszczyk} introduced the Gram-Schmidt walk as an efficient constructive approach to Banaszczyk's vector balancing theorem, obtaining a distribution over colorings with subgaussian discrepancy \citep{294762.294765, dadush2016constructiveversionbanaszczyksvector,dadush2018balancing}. \cite{harshaw2023balancingcovariatesrandomizedexperiments} subsequently established the sharp 1-subgaussian guarantee and developed the walk's application to covariate balance in randomized experiments. These analyses exploit the geometry of the least-squares directions and the conditional mean-zero property of the updates.  In contrast, we study how this guarantee changes under biased, history-dependent perturbations. Such errors can change which coordinate freezes, changing the linear systems solved in future iterations. We analyze the realized perturbed trajectory directly, preserving the guarantee due to the exact-arithmetic in the log moment-generating-function bound but also revealing an additional term due to accumulated error.

\textbf{Algorithmic discrepancy theory.} Several lines of work develop efficient randomized constructions for discrepancy minimization. \cite{alweiss2020discrepancyminimizationselfbalancingwalk} introduced a self-balancing walk with guarantees for online vector balancing \citep{spencer1985six,bansal2020onlinevectorbalancinggeometric,kulkarni2023optimalonlinediscrepancyminimization} and the Komlós problem \citep{spencer1994ten}.  \cite{bansal2022unifiedapproachdiscrepancyminimization} give a stochastic-process framework that recovers several discrepancy bounds, while \cite{deng2025discrepancy} accelerate discrepancy minimization through faster implementations of Edge-Walk and hereditary-discrepancy algorithms. Our focus is the arithmetic accuracy needed to preserve the output concentration of a particular balancing algorithm. A different notion of stability appears in \cite{gamarnik2023geometric}, who establish barriers for algorithms whose outputs change little under perturbations of random input instances. Their input-stability condition concerns sensitivity of the selected coloring; our analysis permits changes in the trajectory and terminal coloring while controlling the resulting discrepancy distribution.

\textbf{Finite precision and numerical stability.} Classical numerical analysis relates roundoff error, backward error, and conditioning for linear-system and least-squares computation \citep{doi:10.1137/1.9780898718027}. Recent works makes precision requirements explicit in algorithms that repeatedly update linear-algebraic quantities. \cite{10.1145/3618260.3649787, 10353098,doi:10.1137/1.9781611978971.188} analyze the backward stability of inverse maintenance and distributed optimization, and derive bit-complexity bounds for regression and linear programming. \cite{diem2011notion,pan1985bit,anand_et_al:LIPIcs.ICALP.2024.10} extend this perspective to dynamic matrix expressions and their determinants. For implementations based on inverse updates, \cite{doi:10.1137/25M1747087} provide related forward- and backward-error bounds for the Sherman-Morrison-Woodbury formula evaluated with approximate inverses.  These results concern the accuracy and cost of numerical subroutines. The Gram-Schmidt Walk additionally requires controlling how these errors interact with adaptive freezing decisions and randomized updates. Our analysis addresses an additional question of  how local numerical errors affect concentration when they also change adaptive freezing decisions and conditional update distributions, connecting the working precision and conditioning to conditional drift, accumulated discrepancy error, and ultimately concentration of the terminal coloring, without requiring exact conditional unbiasedness.\looseness=-1

\textbf{Discrepancy in learning and data compression.} \cite{karnin2019discrepancycoresetssketchesmachine} connect low class discrepancy to small offline and streaming coresets. Particularly close to our setting, \cite{carrell2026lowrankthinning} analyze subgaussian thinning, introduce Gram-Schmidt-based thinning procedures, and give a cubic-time halving implementation based on local inverse updates. Their framework relates concentration of the balancing primitive to compression quality and supports applications to attention approximation, gradient reordering, and two-sample testing. Discrepancy also informs neural-network weight quantization  \citep{chee2025discquantquantizationmethodneural} and attention mechanisms \citep{kochetkova2026streamingattentionapproximationdiscrepancy}. These applications motivate understanding the precision requirements of balancing algorithms. Our contribution provides a concentration guarantee for an implementation of the Gram-Schmidt primitive. In turn, these guarantees are flexible and can be converted into guarantees for repeated compression or any learning pipeline by propagating the resulting error budget through that procedure.

\section{Preliminaries and Notation}

\newcommand{\dth}{\hat{\delta_t}}
\newcommand{\uth}{\hat{u_t}}

We write $[n]$ for $\{1,...,n\}$. For any vector $v \in \R^n$ and $i\in[n]$ we write $v(i)$ for the $i$th entry.
 
We first review a few key properties of subgaussian random variables and random vectors. 

\begin{fact}[Characterization of subgaussian random variables]For a real-valued random variable $X$, the following are equivalent up to constant factors in the parameter $K_i$ \citep{Vershynin_2026}:

\begin{enumerate}
    \item There exists some $K_1$ such that $\Pr[|X| \geq t] \leq 2 \exp \left( -t^2 / K_1 ^2 \right)$ for all $t \geq 0$,
    \item There exists some $K_2$ such that $\E [\exp(X^2 / K_2^2)] \leq 2$,
\end{enumerate}
and for $X$ such that $\E [X] = 0$, the following condition is equivalent to the previous two.
\begin{enumerate}
    \setcounter{enumi}{2}
    \item $\E[\exp (t X)] \leq \exp (K_3 ^2 t^2)$, for all $t \in \mathbb{R}$.
\end{enumerate}

\textup{
A random vector is subgaussian with parameter $K_2$ if all of its unit-direction projections are subgaussian with parameter not exceeding $K_2$, and for the third formulation given we can equivalently show for a random vector $Z$ that $\E[\exp (\inner{Z, v})] \leq \exp(K_3^2 \| v \|_2^2)$ for all vectors $v$, with the norm of $v$ playing the role of the scalar $t$.}
\end{fact}

\textbf{Vector Balancing:}  Suppose we are given $B = [b_1\dots b_n]\in\R^{m\times n}$, where $\|b_i\|_2\leq 1$ for $i\in[n]$, let $z_0 \in [-1,1]^n$. Then, the linear discrepancy objective is to output $z\in\{-1,1\}^n$ such that $\|B(z - z_0)\|_\infty$ is minimized. The Gram-Schmidt Walk samples a coloring whose discrepancy vector has controlled directional tails. This problem has been studied extensively, even in the online setting where only one vector is revealed at a time and signs must be irrevocably assigned \citep{bansal2020onlinediscrepancyminimizationstochastic}.  \cite{harshaw2023balancingcovariatesrandomizedexperiments} shows that $B (z - z_0)$ is $1$-subgaussian (over the randomness of $z$).

\textbf{Gram-Schmidt Walk.} The Gram-Schmidt walk (Algorithm \ref{alg: GS walk original}) achieves its guarantee through an adaptive sequence of randomized updates. At each iteration, it chooses a pivot coordinate, selects a least-squares direction whose pivot entry is one, and moves to one of the two endpoints of the feasible line segment in the unit hypercube $[-1,1]^n$ with probabilities such that the exact update is conditionally mean $0$. Here, coordinates that reach $\{1,-1\}$ are frozen, and the direction is recomputed over the remaining coordinates. The concentration proof of \cite{harshaw2023balancingcovariatesrandomizedexperiments} exploits this probabilistic structure and the orthogonality of the revealed subspaces as the set of possible pivots shrinks.\looseness=-1

\begin{algorithm}[tbp]
    \DontPrintSemicolon
    \small
    \caption{Gram-Schmidt walk \citep{harshaw2023balancingcovariatesrandomizedexperiments}}
    \label{alg: GS walk original}

    Initialize $t \gets 1$ and $z_1 \gets z_0$\;
    Choose $p \in[n]$ uniformly from $A_1:=\{i:|z_1(i)|<1\}$\;
    \BlankLine

    \While{$z_t \notin \{\pm 1\}^n$}{
        Create the active set $A_t \gets \{i\in[n]: |z_t(i)|<1\}$\;
        \If{$p \notin A_t$}{
            Sample $p$ uniformly at random from $A_t$\;
        }
        Let $p_t\coloneqq p$ and let $U_{p_t} \gets \{u\in\R^n: u(p_t)=1,\;
        u(i)=0\text{ for all }i\notin A_t\}$\;
        Set $u_t \gets
        \operatorname*{arg\,min}_{u\in U_{p_t}}\|\mB u\|_2^2$\;
        Let $[-\delta_t^-,\delta_t^+] =
        \{\delta\in\R: z_t+\delta u_t\in[-1,1]^n\}$\;

        Sample $\delta_t$ according to $\delta_t =
            \begin{cases}
                \delta_t^+, & \text{with probability }
                \dfrac{\delta_t^-}{\delta_t^-+\delta_t^+},\\[4pt]
                -\delta_t^-, & \text{with probability }
                \dfrac{\delta_t^+}{\delta_t^-+\delta_t^+}.
            \end{cases}$

        Set $z_{t+1} \gets z_t+\delta_t u_t$, such that $\E[z_{t+1} | z_t] = z_t$\;
        Increment $t \gets t+1$\;
    }
    \Return{$z_t\in\{\pm1\}^n$}\;
\end{algorithm}
The canonical result of \cite{harshaw2023balancingcovariatesrandomizedexperiments} proves that the outputs are $1$-subgaussian.
\begin{theorem}[\citep{harshaw2023balancingcovariatesrandomizedexperiments}]
Let $\mB=[b_1 \ldots b_n]\in\mathbb R^{m\times n}$ satisfy
$\|b_i\|_2\le1$, and let $z_0\in[-1,1]^n$. Then, for all $v\in\R^m$, the terminal output $z$ of the walk in exact arithmetic satisfies
\begin{equation}\mathbb E\exp\bigl(\langle B(z-z_0),v\rangle\bigr)
\le\exp\left(\frac{1}{2} \|v\|_2^2 \right).\end{equation}
\end{theorem}

\subsection{Technical Contributions}

We analyze the algorithm under a refined model that takes into account errors in calculating and sampling the step at each iteration, while still remaining general enough to apply across machine configurations and least squares solvers. Our main theorem parameterizes the key sources of error
at each step. Appendix~\ref{sec: missing proofs} then shows sufficient
working-precision conditions for a specified implementation,
under explicit assumptions on least-squares stability,
conditioning, and exponent range. In particular we consider the error of the least squares solve $\els$, the max coordinate error due to approximate arithmetic in computing the step and clamping $z_t$ back into the hypercube $\ec$,  and the perturbation in the distribution of $\dth$ which we denote $\epsilon_{pr}$. We discuss our model in more detail in Section \ref{error_sources}.\looseness=-1

\begin{figure}[t]
    \centering
\includegraphics[width=0.7\linewidth]{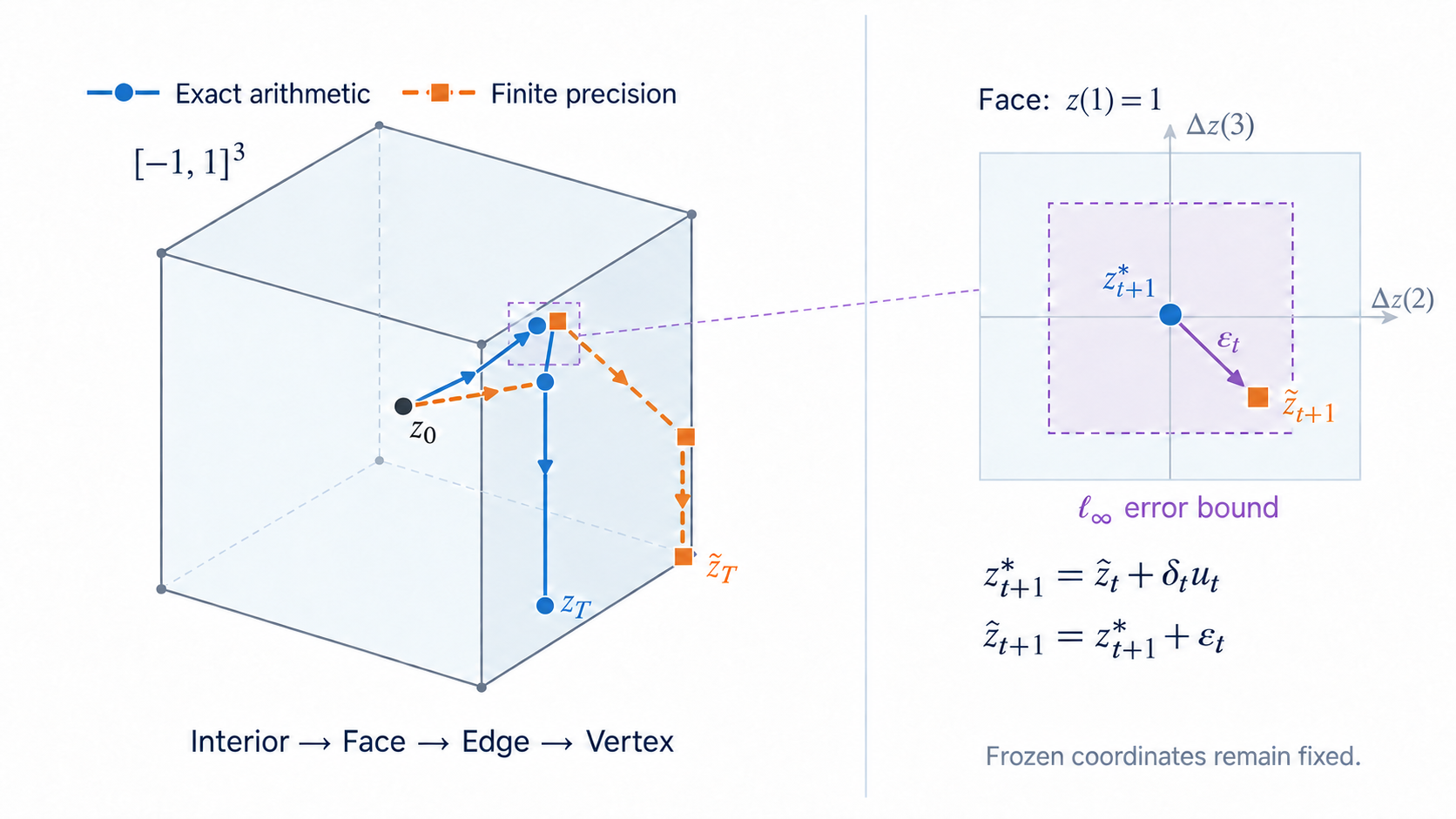}
    \caption{Finite-precision perturbations in the Gram-Schmidt walk. Left: Exact-arithmetic (solid blue) and finite-precision (dashed orange) walks start from the same point but can follow very different trajectories as coordinates successively freeze. Right: Each finite-precision update differs from the ideal update computed at its current state by a local error due to our tracked error sources.}
\end{figure}

Note that the introduction of errors due to rounding can affect the path that the Gram-Schmidt walk may take, so the resultant walk may deviate significantly from the original ideal walk. Therefore, simply bounding the deviation from the original walk can lead to very weak guarantees. So, we treat this walk as its own object when bounding its subgaussianity. We show the following bound:

\begin{theorem} \label{thm:spielmanadapted}
    Under the model and assumptions of Section \ref{sec finite precision gs walk}, the output vector $B(z - z_0)$ satisfies the following moment generating function (MGF) bound: for all $v\in\R^m$,
    \begin{equation*}
        \E [\exp(\inner{B(z - z_0), v})]\!\leq\!
        \exp\!\left(\!\frac{1}{2} \norm{v}^2\!+\! \left(\underbrace{n^{3/2}(2 \epsilon_{pr} \!+\! \ec)}_{\text{drift/sampling}}\!+\!\! \left(\underbrace{2 n \els}_{\text{least squares}}  \!\!+\underbrace{n(n+1) \ec}_{\text{arithmetic/clamping}}\right)\!\!\right)\!\!\norm{v}\!\right).
    \end{equation*}
\end{theorem}

We provide values for these constants depending on machine configuration in Appendix~\ref{sec: missing proofs}. For example, if we assume a floating point machine of unit roundoff $u$, then as discussed in Lemma \ref{consants:prob_error} we can show that $\epsilon_{pr}$ scales according to the unit roundoff. Furthermore Lemma \ref{corollary: ls-backward-error} discusses how under special conditions, namely a backwards stable solver with full column rank $B$, we expect $\els$ to scale with condition number and unit roundoff. This is meant to give an idea for how $\els$ might scale in some special case rather than to hold in general. Finally, lemma~\ref{lem:update_precision} bounds arithmetic and freezing corrections
for a specified implementation. Combining these results,
Corollary~\ref{cor:sufficient_bits} gives sufficient significand and sampling
precision for a prescribed additive discrepancy penalty.

Finite-precision errors obstruct a direct application of the analysis in \cite{harshaw2023balancingcovariatesrandomizedexperiments}: they can change the freezing trajectory and destroy the conditional centering on which the phase argument relies. We overcome this by proving a phase MGF bound to explicitly account for accumulated drift along the realized perturbed trajectory, where the key step is to incorporate the pivot-coordinate update error into the scalar step so that its sum over a phase is the actual pivot displacement. This preserves the structure of the terminal displacement while allowing biased conditional probabilities, whose effect we control via a modified backward-induction argument. Combining these drift bounds across orthogonal phase subspaces produces the $n^{3/2}(2\epsilon_{\mathrm{pr}}+\epsilon_{\mathrm{iter}})$ contribution. The remaining terms follow from pathwise bounds on least-squares and coordinate-update error, including clamping and freezing corrections; hence, our bound keeps the original quadratic MGF term while adding an explicit linear penalty in $\norm{v}$ that vanishes with error to recover the exact-arithmetic guarantee.\looseness=-1

We note that this is not a subgaussianity guarantee in itself due to the linear term. Nevertheless, we show that an MGF bound of this form implies subgaussianity in the Orlicz norm formulation giving a subgaussianity result: Lemma \ref{lem:mgf_orlicz} in the appendix implies
\begin{equation*}
    \norm{B(z_{}-z_0)}_{\psi_2} \leq \sqrt{6} + \frac{1}{\sqrt{\log 2}}\left(n^{3/2}(2 \epsilon_{pr} + \ec) + (2 n \els  +n(n+1) \ec )\right)
\end{equation*}
and so by the equivalence of the forms of Subgaussianity up to constants we also get Gaussian-like exponential tails on $\abs{B (z - z_0), v}$ for every unit vector $v$.

\begin{theorem}[Lower bound for bounded update errors]
\label{thm:bounded-error-lower-bound}
There exists a universal constant $c>0$ such that, for every integer $n\ge 8$ and every $\varepsilon\ge 0$, there is a full-column-rank matrix $\mB=[b_1,\ldots,b_n]\in\mathbb{R}^{(n+1)\times n}$ with $\|b_i\|_2=1$ and an admissible error process for the Gram-Schmidt Walk, initialized at $z_0=0$, such that $z_{t+1}=z_t+\delta_tu_t+\epsilon_t$ with $\|\epsilon_t\|_\infty\le\varepsilon$, with exact step-selection probabilities, and whose terminal discrepancy $D=\mB z_T$ satisfies 
\[\|\mathbb{E} D\|_2 \ge c\min\{n^2\varepsilon,n\}, \qquad\|D\|_{\psi_2} \ge c\min\{n^2\varepsilon,n\}.\] Here $\|\cdot\|_{\psi_2}$ denotes the uncentered directional Orlicz norm. The errors are deterministic functions of the
iteration, active set, and current pivot, which preserve all
frozen coordinates, and the walk remains feasible and
terminates after at most $n$ iterations.
\end{theorem}

The construction proves that the $n^2 \ec$ dependence is unavoidable, up to universal constants, for worst-case bias under admissible bounded coordinate errors, even when the least-squares directions and branch probabilities are exact. The same construction lower-bounds the uncentered directional Orlicz norm. This saturation at order $n$ is consistent with $\|Bz_T\|_2\le\sum_{i=1}^n\|b_i\|_2\le n$ when $z_0=0$.

We also numerically\footnote{We open-source the code for our experiments in the following GitHub repository: \url{https://github.com/emiletimothy/FinitePrecisionGramSchmidtWalk}.} study how arithmetic precision and problem size
affect discrepancy, where we illustrate the predicted quadratic accumulation
of bounded update errors in Figure \ref{fig:collapse}.

\section{Finite-Precision Gram-Schmidt Walk}
\label{sec finite precision gs walk}
Because of the numeric issues outlined above, the existing proofs bounding the discrepancy of the Gram-Schmidt walk \cite{bansal2017gramschmidtwalkcurebanaszczyk,harshaw2023balancingcovariatesrandomizedexperiments} no longer hold.
The proof in \cite{bansal2017gramschmidtwalkcurebanaszczyk} uses Martingales, but these properties are lost due to the errors.
The more recent proof from \cite{harshaw2023balancingcovariatesrandomizedexperiments} builds on the fact that the exact solution to the L2-minimization problem can be expressed via results from an exact Gram Schmidt orthogonalization procedure. However, this too no longer holds in the presence of errors. Now, we explain the main technical ideas and tools needed to robustify the proof from \cite{harshaw2023balancingcovariatesrandomizedexperiments} against numerical issues.\looseness=-1

\paragraph{Error Sources in Finite-Precision Model} \label{error_sources}

At the realized state $z_t$, let $u_t$ denote an exact least-squares direction and let $\widehat u_t$ denote the computed direction. We state our bounds relative to the realized input. In our framework, we enforce $\widehat u_t(p_t)=1$ and $\widehat u_t(i)=0$ for $i\notin A_t$, and assume \begin{equation} \|B(\widehat u_t-u_t)\|_2 \leq \els. \end{equation}
All errors from the final arithmetic update, clamping, and explicit freezing are collected in $\clamp$, so that $z_{t+1} = z_t+\dth\,\widehat u_t+\epsilon_t$. The update preserves feasibility and previously frozen coordinates, freezes at least one additional coordinate, and retains the current pivot until it freezes. Let $\widehat\delta_t^+,\widehat\delta_t^->0$ denote the
computed endpoint magnitudes, so that
$\dth\in\{\widehat\delta_t^+,-\widehat\delta_t^-\}$. For convenience, we cap these magnitudes by downward-rounded
pivot distances, ensuring $\widehat\delta_t^+\leq1-z_t(p_t)$ and $\widehat\delta_t^-\leq1+z_t(p_t)$. Thus,
$\widehat\delta_t^++\widehat\delta_t^-\leq2$ and
$|\dth|\leq2$.
Conditioning on the history and all computations preceding the random branch choice, let $p^{\mathrm{comp}} \coloneqq \frac{\widehat\delta_t^-} {\widehat\delta_t^++\widehat\delta_t^-}$ be the probability that centers the step at the computed endpoints. We assume that the actual positive-branch probability $\widehat p$ satisfies $|\widehat p-p^{\mathrm{comp}}|\leq\varepsilon_{pr}$, so 
\begin{equation} 
\left|\mathbb E[\dth\mid\mathcal F_t]\right| 
= 
(\widehat\delta_t^++\widehat\delta_t^-) |\widehat p-p^{\mathrm{comp}}| \leq2\varepsilon_{pr}. 
\end{equation}

In this convention, the discrepancy increment decomposes as 
$B(z_{t+1}-z_t) = \dth Bu_t +\dth B(\widehat u_t-u_t) +B\epsilon_t$,
where the least-squares contribution $\dth B(\uth - u_t)$ has norm at most $2\els$.

\paragraph{Analysis framework}

We sketch the main ideas of \citet{harshaw2023balancingcovariatesrandomizedexperiments}
and explain which pieces break due to numerical issues.
We then later explain how to fix them.

Let $\sigma:[n]\to[n]$ be the reverse order in which the Gram-Schmidt Walk decides coordinates $i$ (that is, when $i$ is picked as a pivot or $z(i)$ becomes $\pm1$).
Now assume we run Gram Schmidt orthogonalization on the columns of $B$ in order $\sigma$ and let $w_{\sigma(1)},...,w_{\sigma(n)} \in \R^m$ be the orthonormal vectors produced.
We allow for $B$ to not be full column rank and so the Gram-Schmidt orthogonalization procedure may produce a zero vector in the sequence; we refer to sets as orthonormal throughout our exposition with the understanding that we refer to the nonzero vectors being discussed.
The vectors $w_{\sigma(i)}$ are not computed and used only for analysis purposes.

We consider the progression of the walk in pivot phases, where $t_p$ is the iteration at which $p\in[n]$ is first selected as pivot.
Within a phase, $b_p$ refers to the current pivot column of $B$.
When analyzing iteration $s \ge t_p$ of the phase for a pivot $p$, let $Q_{s,p}$ be the set of positions $r$ decided in iteration $s$ during the pivot phase of $p$. So $Q_{s,p}=\emptyset$ for $s \ge$ the end of pivot $p$'s phase.
While pivot $p$ is technically decided at the start of iteration $t_p$, we charge it to iteration $t_p-1$, i.e., $Q_{t_p - 1,p}$ consists of just the pivot $p$.
The proof of \cite{harshaw2023balancingcovariatesrandomizedexperiments} heavily relies on the following decomposition:

\begin{lemma}[Lemma $S3.2$ in \cite{harshaw2023balancingcovariatesrandomizedexperiments}]
\label{spielman:dirstepisproj}
    At each iteration $t$ we have for pivot $p=p_t$
    \begin{equation*}
        B u_t = \sum_{s=t_p-1}^{t-1} \sum_{r \in Q_{s}} \inner{w_{\sigma(r)}, b_p} w_{\sigma(r)}.
    \end{equation*}
\end{lemma}

This is only true for the exact solution $u_t$ and does not hold for the approximation $\uth$ that would be produced by numeric implementation of the algorithm.

Finally, consider the final vector $B (z_T - z_0) = B(\sum_t \delta_t u_t)$ produced in exact arithmetic. Using \Cref {spielman:dirstepisproj}, \citet{harshaw2023balancingcovariatesrandomizedexperiments} proves that for any $v\in\R^m$
\begin{align}
    \E\left[  \exp \left( \sum_{t=1}^n \delta_t \inner{B u_t, v} \right) \right] \leq \exp \left(\frac{1}{2} \norm{v}^2\right) \label{eq:harshaw}
\end{align}
This is done by splitting the run of the algorithm into phases based on the pivots. For a pivot $p$, let $S_p \subset [n]$ be the iterations where $p$ was pivot, and let $\Delta_p$ be all the randomness up to the point that $p$ is chosen as pivot.
Then \cite{harshaw2023balancingcovariatesrandomizedexperiments}
proves 
\begin{equation}
\E\left[\exp \left( \sum_{t\in S_p} \delta_t \inner{B u_t, v} -\frac{1}{2} \|P_p b_p\|^2 \|P_p v\|^2\right)  \bigg| \Delta_p \right] \le 1 \label{eq:harshaw2}
\end{equation} 
which can then be assembled to prove \eqref{eq:harshaw}, with the $\|P_p b_p\|^2 \|P_p v\|^2$ are orthogonal projection terms accumulating to at most $\norm{v}^2$.
However, for $\uth$ and $\dth$ the conditional expectation is no longer bounded by 1 because $\E[\dth]\neq 0$ and \Cref{spielman:dirstepisproj} does not hold for $\uth$.

\textbf{Bounding the error.} \label{paragraph:error_decomp}
Consider the final vector $B (z - z_0) = B (\sum_t \dth \uth + \epsilon_t)$. 

We rearrange the updates as follows:
\begin{align}
Bz - Bz_0 &= B \left( \sum_{t=1}^T \dth u_t + \dth(\uth - u_t) + \clamp \right)  \notag\\
&= B \sum_{t=1}^T \left(\dth + \epsilon_t(p_t)\right) u_t + B \left(\sum_{t=1}^T 
\dth (\uth - u_t)+ \epsilon_t - \epsilon_t(p_t) u_t
\right) \label{eq:overview:split}
\end{align}
Here the latter term can be bounded by using the definition of our errors, giving the bound
$$
\left\| \sum_{t=1}^T B ( 
\dth (\uth - u_t)+ \epsilon_t - \epsilon_t(p_t) u_t) \right\|_2
\le
2 n \els  +n(n+1) \ec.$$
The calculation is given in \Cref{sec: pf of theorem 2.3}.

The first term in \eqref{eq:overview:split} is similar to \eqref{eq:harshaw} in that it uses the exact $u_t$, however, the probability $\dth$ is still inexact.
As a result, we prove the following weaker bound
with a variable $\gamma$ upper bounding the drift.

\newcommand{\EE}[2]{\mathbb{E} \left[\exp \left( #1 \right) \,\middle|\, #2 \right]}
\newcommand{\EEr}[1]{\mathbb{E} \left[\exp \left( #1 \right) \,\middle|\, \Delta_R \right]}
\newcommand{\EEe}[1]{\mathbb{E} \left[\exp \left( #1 \right) \right]}
\newcommand{\greet}[1]{Hello, #1!}

\begin{lemma}\label{spielman:withinphasebound}
For $\gamma=n(2 \epsilon_{pr} + \ec)$, the non-noise portion of the updates satisfy the MGF bound 
\begin{equation*}
    \EEe{\sum_{t=1}^n (\dth + \epsilon_t (p_t)) \inner{B u_t, v}} \leq \exp (\frac{1}{2} \norm{v}^2 + \sqrt{n} \gamma \norm{v} )
\end{equation*}
\end{lemma}

The main difficulty in adapting the proof of \eqref{eq:harshaw} by \cite{harshaw2023balancingcovariatesrandomizedexperiments} to \Cref{spielman:withinphasebound},
is the reliance on the distribution of $\sum_t \delta_t$ within a pivot phase being a two point distribution, which no longer holds due to noise along the pivot unit. 
The high-level idea of our proof is that, while \eqref{eq:harshaw2} does not hold, it does hold after adding some penalty term $E_t$ such that
$$
\E\left[\exp\left(\sum_{t\in S_p} \delta_t \inner{B u_t, v} -\frac{1}{2} \|P_p b_p\|^2 \|P_p v\|^2 - E_t\right) \bigg | \Delta_p\right] \le 1
$$
The main difficulty is to find a large enough $E_t$ for which the inequality holds, but not too large to get a bound weaker than \Cref{spielman:withinphasebound}.
The full proof is long and can be found in \Cref{sec:singlepivot,appendix:constants}, but it turns out that
$E_t = \gamma \sum_{s\ge t_p-1} \sum_{r\in Q_{s,p}} |\alpha_r\beta_r|$
for $\gamma = n(2 \epsilon_{pr} + \ec)$,
$\alpha_r=\inner{w_{\sigma(r)}, b_p}$ and $\beta_r = \inner{w_{\sigma(r)}, v}$ is large enough.
We further have
\begin{equation*}
    \sum_{p} \sum_{s\ge t_p-1}^{n} \sum_{r \in Q_{s,p}} \abs{\alpha_r \beta_r}
    \leq \Bigl(\sum_{p=1}^n \norm{P_p b_p}^2\Bigr)^{1/2} \Bigl(\sum_{p=1}^n \norm{P_p v}^2\Bigr)^{1/2}
    \leq \sqrt{n}\, \norm{v},
\end{equation*}
where $P_p$ is the orthogonal projection matrix onto the $w_{\sigma(r)}$ decided in pivot phase $p$.
This yields the additional term in \Cref{spielman:withinphasebound}.
Assembling these two bounds on \eqref{eq:overview:split} then yields \Cref{thm:spielmanadapted}.

\section{Numerical Experiments}

\textbf{Compute.} Our walks run in parallel; we use a process pool with per-worker
\texttt{SeedSequence} streams and single-threaded BLAS. The per-operation
rounding makes the low-precision path $\sim\!20\times$ slower than float64's. Our main figures use $1000$ walks per cell on a 6-core laptop (1-3 CPU-hours).

\textbf{Matrices.} We use four matrix ensembles.
\begin{enumerate}[label=(\roman*),leftmargin=1.8em,itemsep=2pt,topsep=2pt]
  \item \emph{Clustered.} $B = u\mathbf{1}^{\top} + m^{-1/2}G$, with $u$ a
  random unit vector and $G$ Gaussian, column-normalized. The columns
  concentrate near a shared direction, which makes $B_{A \setminus p}$
  increasingly ill-conditioned as the active set shrinks through $k \approx m$,
  \item \emph{Higgs.} The two low-level angular-momentum features of the Higgs
  two-sample benchmark used in low-rank thinning \citep{carrell2026lowrankthinning}, giving a $2 \times n$
  matrix,
  \item \emph{Identity.} $B = I_{n}$, a tight control in which $d^{\top}Bz$ is
  a Rademacher sum with $\sigma = 1$.
  \item \emph{Adversarial lower bound.} Columns
  $v_{i} = \frac{e_{0} + e_{i}}{\sqrt{2}}\in\R^{n+1}$, where a
  bounded per-step error $a$ applied to non-pivot coordinates in the
  first $\frac{n}{4}$ iterations forces an $e_{0}$-discrepancy of order $\Omega(n^{2} \epsilon)$.\looseness=-1
\end{enumerate}

\textbf{Arithmetic model.} Our main sweeps emulate reduced precision in float64: every elementary operation rounds its significand to $b$ retained fraction bits (round-to-nearest), where $b$ is the number of stored fraction bits. Here, we use the format labels fp8-E5M2, bf16, and fp16 to denote the matching fraction width ($b=2,7,10$). Similarly, reductions (inner products, norms, matvecs) round each product but accumulate in float64, rounding the result once. \begin{wrapfigure}[18]{r}{0.40\textwidth}
 \centering
 \vspace{-15pt}
\includegraphics[width=\linewidth]{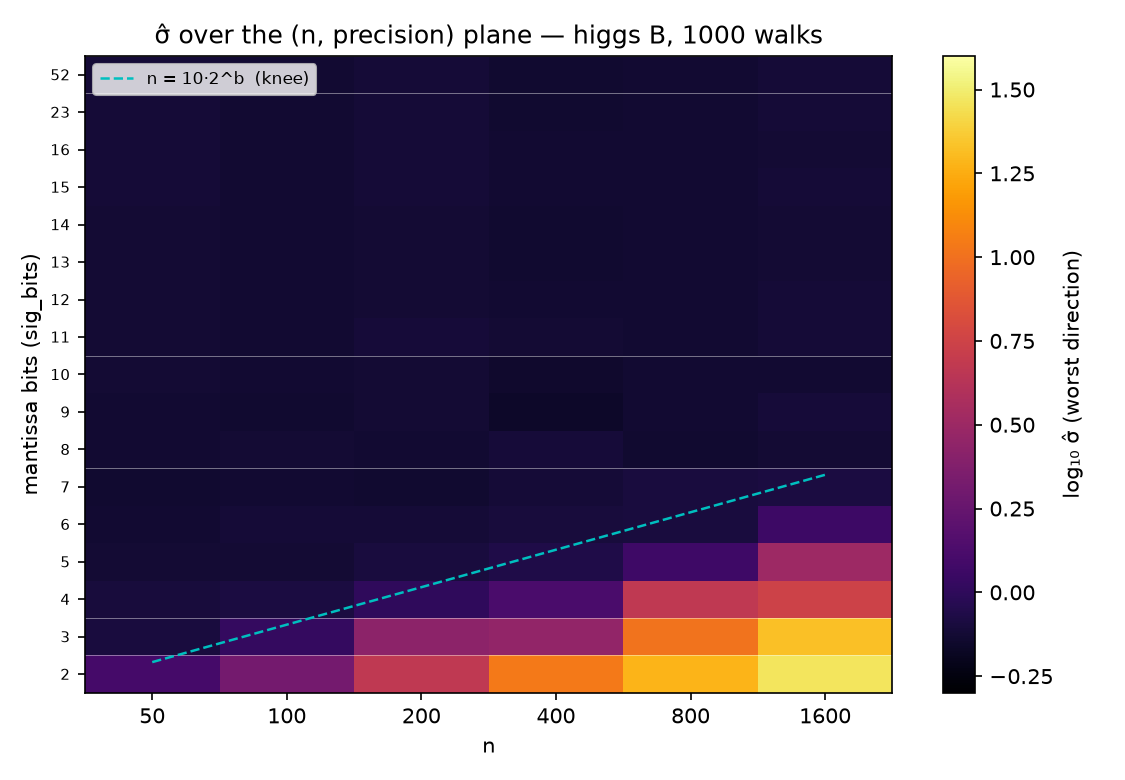}
  \caption{{\textbf{Estimated subgaussianity over the $(n, \mathrm{bits})$ plane.}} Each
  cell is $\log_{10}\hat{\sigma}$, where the direction is selected via independent walks on the Higgs matrix.
  The float64 reference plateau goes to 8 mantissa bits at small $n$ and
  ends along the curve $n=10 \cdot 2^{b}$. Above it,
  $\hat{\sigma}$ grows along lines of constant $n \cdot 2^{-b}$.}
  \label{fig:phase}
  \label{figure: heatmap}
\end{wrapfigure} The least-squares solve reads the $b$-bit rounding $\hat B = \mathrm{fl}(B_0)$ of the input, corresponding to $L=b+1$ significand bits and unit roundoff $u=2^{-(b+1)}$ bits, while the experiments of Appendix \ref{subsection: dtype comparisons} use the actual dtypes.
 
 \textbf{Metrics.} For each main-sweep configuration, we run $1000$ independent walks. \emph{Discrepancy.} We report the sample mean of $\norm{Bz}_1/m$ and the empirical CCDF of $\norm{Bz}_\infty$ on squared-threshold, log-survival axes. \emph{Directional concentration.} Across our sweeps, we use half of the
walks to select a direction maximizing a centered moment
diagnostic based on orders $2,4,6$. Using the remaining
walks, we evaluate this diagnostic, the empirical
$\psi_2$ norm, and centered exceedance frequencies in
the selected direction. Error bars describe five-batch variability as detailed in Appendix ~\ref{sec: experimental details}. \emph{Solver error.} We follow a float64 reference trajectory and recompute directions in low precision at identical states, recording $\Delta u_t=u_{\mathrm{lp}}-u_{\mathrm{fp64}}$ by the number $k$ of free nonpivot indices. We summarize these coefficient errors using their median, interquartile range, and mean.\looseness=-1

\begin{figure}[t]
  \centering
  \includegraphics[width=\linewidth]{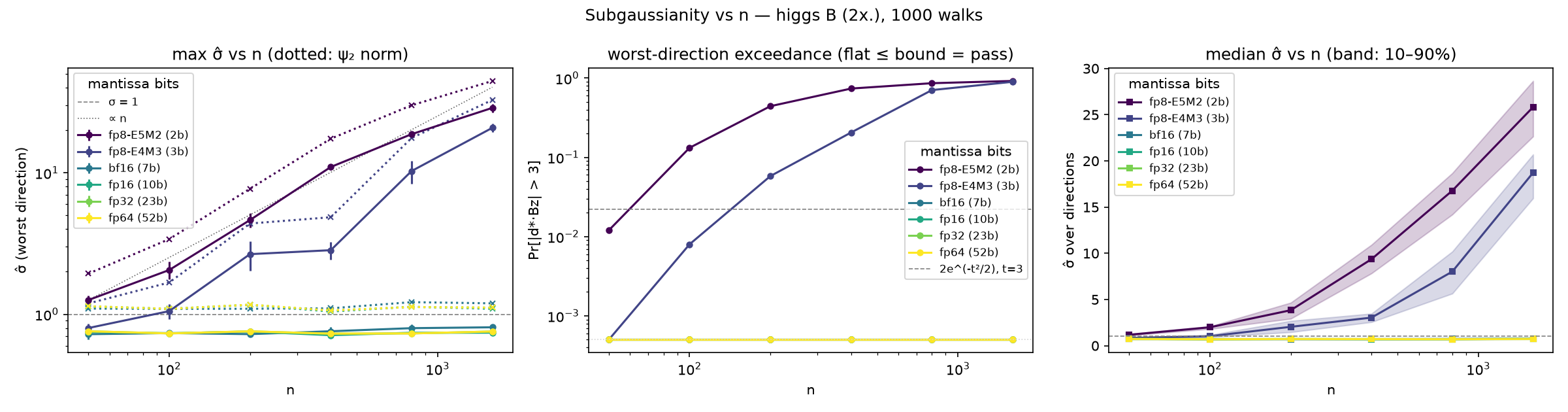}
  \caption{\textbf{Subgaussianity vs.\ $n$ across fp formats on the Higgs
  matrix.} Left: the centered moment diagnostic  and
empirical $\psi_2$ norm, evaluated in a direction
selected using an independent half of the walks.
The lowest precisions exhibit substantial growth with
$n$, while higher-precision curves show much weaker
dependence on $n$.
Middle: evaluation-set frequencies of
$|d_\star^\top B_0z-\overline X_{\mathrm{ev}}(d_\star)|>3$;
the dashed line $2e^{-9/2}$ is the ideal unit-subgaussian
reference. 
Right: the median centered moment diagnostic over
tested directions, with a $10$--$90\%$ directional band.
The rising median at the lowest precisions shows
deterioration across a substantial portion of the
tested directions.}
  \label{fig:nsubgauss}
\end{figure}

\begin{figure}[t]
  \centering  \includegraphics[width=0.8\linewidth]{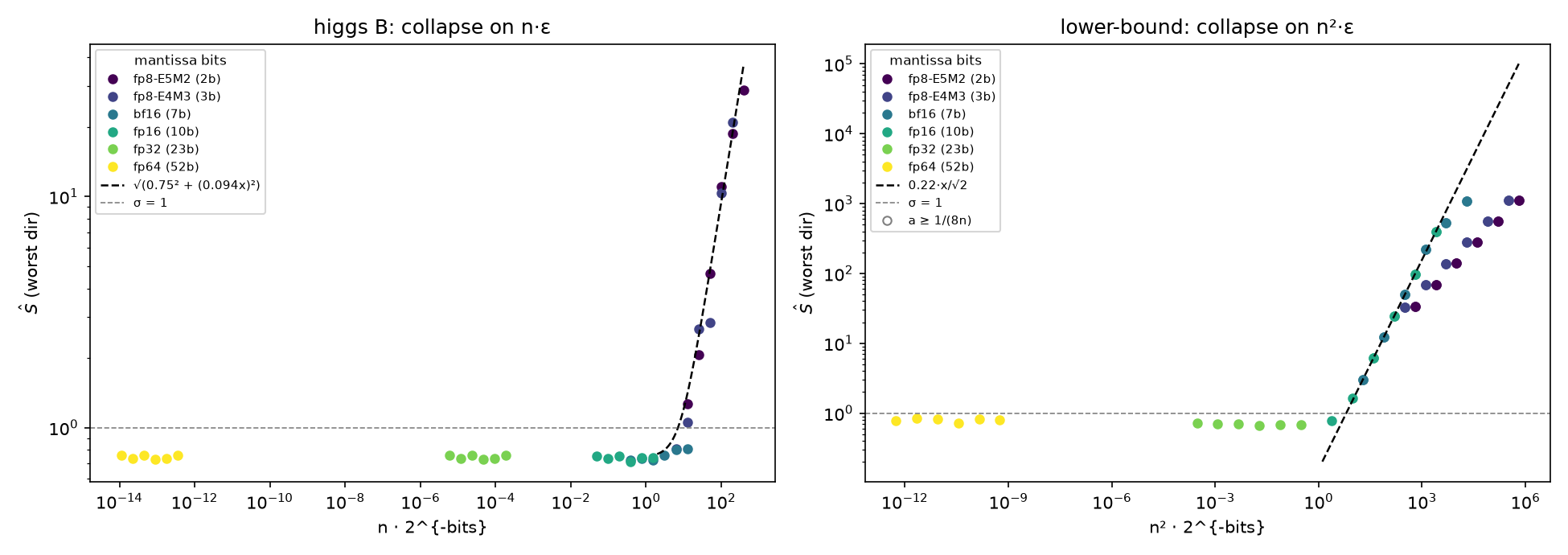}
  \caption{\textbf{Scaling collapse.} Left: on the Higgs matrix, $\hat{S}$ is
described by a single curve $\sqrt{\sigma_0^2 + (c\, n\, 2^{-b})^2}$ in
$n\cdot 2^{-b}$. Right: on the lower-bound matrix, $\hat{S}$ follows
$n^2 \cdot 2^{-b}$ (dashed) in the degraded regime and departs from it where
the bound $\min\{n^2\varepsilon, n\}$ saturates near $n/\sqrt{2}$ or falls to
the float64 plateau; open markers violate $a < 1/(8n)$.\looseness=-1}
  \label{fig:collapse}
\end{figure}

\textbf{Findings.} At fixed precision, our experiments show that the discrepancy inflation generally grows with the number of vectors $n$, with the strongest deterioration occurring at the lowest precisions. 
Increasing the number of retained bits mitigates this growth and brings the empirical discrepancy statistics closer to the float64 reference, suggesting that larger instances need higher precision to maintain comparable performance.
 
\textbf{Implementation details.} We compare fp16/fp32 implementations of Householder QR least squares, the \textsc{GSWDesign} Cholesky/Woodbury scheme, and the explicit-inverse Sherman-Morrison scheme from low-rank thinning's cubic walk, using common random draws. On the tested instances, QR is the most robust, with median image error less sensitive to conditioning than coefficient error. 
Refactoring every ten steps reduces fp16 Cholesky downdate failures to at most $1.75\%$, but $\psi_2$ inflation persists among successful runs. The explicit-inverse scheme is more sensitive to ill-conditioning, matching 
 the squared conditioning of Gram matrices where $\kappa_2(B_S^\top B_S)=\kappa_2(B_S)^2$ of full-rank $B_S$.\looseness=-1

\vspace{-0.2cm}
\section{Conclusion}
We analyze the Gram-Schmidt Walk under bounded, potentially biased and history-dependent perturbations. We show a concentration bound that separates least-squares image error, coordinatewise update error, and conditional sampling error while retaining the exact-arithmetic quadratic term in the log-MGF. We also give a full-column-rank construction that produces a bias of order $\min\{n^2\ec,n\}$, which shows that a quadratic accumulation of update error is unavoidable before saturation. Finally, our experiments corroborate our theoretical findings on how numerical discrepancy varies with fixed precision and problem size. 

\vspace{-0.1cm}

\section*{Acknowledgements}
We thank Abhishek Shetty, Albert Weng, Daniel Zhang, and Mehrdad Ghadiri for insightful discussions on this work. Emile Anand is supported by NSF Grants CCF 2338816, and Jan van den Brand is supported by NSF Grants CCF 2338816 and CCF 2504994.

\vspace{-0.1cm}

\subsection*{AI Disclosure}

In this work, we disclose the use of Cognition's Devin agent to help run our suite of experiments, and we used Anthropic's Claude Fable 5.1 model to suggest improvements on an earlier version of the theoretical result. We have reviewed all AI-assisted work. LLM-generated code was verified and tested for correctness, and we take responsibility for all the final content of this work.

\vspace{-0.1cm}

\subsection*{Ethics statement}

This work studies the worst-case numerical reliability of the Gram-Schmidt Walk, an algorithm for vector balancing. By quantifying the effects of finite precision, our results support more reliable use of discrepancy-based methods and clarify circumstances in which numerical errors can compromise output quality. Applications involving experimental design or data compression require separate consideration of data privacy, representativeness, and the consequences of the resulting decisions.

\vspace{-0.1cm}

\subsection*{Reproducibility statement}
The theoretical assumptions and perturbation models underlying our theoretical guarantees are specified in the main body, with proofs provided in the appendices. Our experimental setup describes our instance generation, dataset preprocessing, solver and factorization-maintenance procedures, numerical precision settings, random seeds, and evaluation metrics. To support reproduction of the finite-precision experiments, we document the arithmetic and accumulation precisions, explain our compute setup, and provide our code with thorough documentation. 

\bibliographystyle{plainnat}
\bibliography{iclr2027_conference}

\newpage

\appendix

\textbf{Outline of the Appendix}
\begin{itemize}
    \item Appendix~\ref{sec: auxiliary lemmas} collects auxiliary lemmas used in our concentration analysis.
    \item Appendix~\ref{sec:singlepivot} gives an analysis of the single pivot phases.
    \item Appendix~\ref{sec: pf of theorem 2.3} proves Theorem~2.3 by combining the phase concentration estimates with bounds on the accumulated numerical errors.
    \item Appendix~\ref{sec: lower bound} presents our lower-bound construction and establishes the worst-case bias caused by bounded update errors.
    \item Appendix~\ref{sec: experimental details} describes our experimental setup and provides additional numerical results.
\end{itemize}

\section{Auxiliary Lemmas}
\label{sec: auxiliary lemmas}
 We now state a lemma from \cite{harshaw2023balancingcovariatesrandomizedexperiments} to be used in our adapted proof, to allow us to bound the expectation of a certain MGF quantity.
\begin{lemma}[Lemma $S3.5$ in \cite{harshaw2023balancingcovariatesrandomizedexperiments}]\label{spielman:operatorfn}
    For all $x \in [-1,1]$ the function
    \begin{equation*}
        f_x \begin{bmatrix} \alpha & \eta \\ \eta & \beta \end{bmatrix}
        = \exp \left(-\frac{1}{2} \alpha \beta\right) \left[\frac{1+x}{2} \exp ((1-x) \eta) + \frac{1-x}{2} \exp ((-1-x) \eta)\right]
    \end{equation*}
is monotone decreasing with respect to the Loewner order on the cone
of $2\times2$ positive semidefinite matrices: if $M\preceq N$, then
$f_x(M)\ge f_x(N)$.\\
\end{lemma}

\begin{lemma}[MGF Bound to Orlicz Norm] \label{lem:mgf_orlicz} Let $a\geq 0$.
If a random vector $X \in \R^m$ satisfies an MGF bound of the form
\begin{equation*}
    \E [\exp (\inner{X, v})] \leq \exp \left(\frac{1}{2} \norm{v}^2 \sigma^2 + a \norm{v} \right)
\end{equation*}
for all $v \in \R^m$ and some $\sigma > 0$ then we have that $\norm{X}_{\psi_2} \leq a/\sqrt{\ln 2} + \sqrt{6} \sigma$.
\end{lemma}
\begin{proof}
Recalling that we define the Orlicz norm of a random vector as the supremum of the Orlicz norm over its one dimensional scalar projections, fix some $u \in S^{m-1}$ and let $Y := \inner{X, u}$. 
We will ultimately show that $\E \left[\exp\left(Y^2 / (a/\sqrt{\ln 2} + \sqrt{6} \sigma)^2\right)\right] \leq 2$. 
Observe that because $X$ satisfies the above MGF bound we have, for any $\lambda \geq 0$,
\begin{align*}
    \E [e^{\lambda Y}] = \E[\exp \inner{X, \lambda u}] \leq \exp \left( \frac{1}{2} \lambda^2 \sigma^2 + a \lambda \right)
\end{align*}
and therefore by applying Markov's inequality in the MGF we get
\begin{align*}
    \Pr [\abs{Y} \geq t] = \Pr [e^{\lambda \abs{Y}} \geq e^{\lambda t}] \leq 2 \exp \left(\frac{1}{2} \lambda^2 \sigma^2 + a \lambda - t \lambda \right)
\end{align*}
and choose $\lambda = \frac{t-a}{\sigma^2}$ to get the following bound for $t \geq a$:
\begin{align*}
    \leq 2\exp(-(t-a)^2 / (2 \sigma ^2))
\end{align*}
From here we define $Z := (\abs{Y} -a )_+$ such that $Z \geq 0$. Observe that for $r > 0$ we have 
\begin{equation*}
    \Pr [\abs{Z} \geq r] = \Pr [(\abs{Y}-a)_+ \geq r] \leq \Pr [\abs{Y} \geq a + r] \leq 2 \exp (\frac{-r^2}{2 \sigma^2})
\end{equation*}
which allows us to show the claimed bound on the MGF of $Y^2$ via an application of the layer cake representation to first bound the MGF of $Z^2$. Let $K > 0$ be chosen later, considering the function $f(r) = e^{r^2/K^2}$ we have
\begin{align*}
    \E [\exp (Z^2 / K^2)] &= 1 + \int_0^\infty f' (r) \Pr [r < Z] d r \\
    &= 1 + \int_0^\infty \frac{2r}{K^2} e^{r^2/K^2} \Pr [r < Z] d r \\
    &\leq 1 + \int_0^\infty \frac{2r}{K^2} e^{r^2/K^2} \cdot 2 \exp (-r^2 / (2 \sigma^2)) d r \\
    &= 1 + \int_0^\infty \frac{4r}{K^2} \exp \left(-r^2 \left(\frac{1}{2 \sigma^2} - \frac{1}{K^2} \right) \right) d r \\
    &= 1 + \frac{4 \sigma^2}{K^2 - 2 \sigma ^2},
\end{align*}
where the final step holds for $K^2 > 2 \sigma ^2$. To make this at most $2$ we set $K = \sqrt{6} \sigma$ which makes the above expression equal $2$, ensuring convergence of the above integral. This gives  $\norm{Z}_{\psi_2} \leq \sqrt{6}\sigma $.

To conclude the proof we use that $\abs{Y} \leq Z + a$ pointwise implying $\E[e^{Y^2}] \leq \E[e^{(Z + a)^2}]$ and so \[\norm{Y}_{\psi_2} \leq \norm{Z + a}_{\psi_2} \leq \norm{Z}_{\psi_2} + \norm{a}_{\psi_2} \leq \sqrt{6} \sigma + a/\sqrt{\ln 2}\]
where we used the triangle inequality followed by the boundedness of $a$ in the final two steps.\qedhere
\end{proof}

\section{Missing Proofs}
\label{sec: missing proofs}
\subsection{Single Pivot Phase Analysis (Proof of Lemma 3.2)}
\label{sec:singlepivot}

Our proof follows the same structure as in \cite{harshaw2023balancingcovariatesrandomizedexperiments}.
To prove \Cref{spielman:withinphasebound}, which states
\begin{align}
\EEe{\sum_{t=1}^n (\delta_t + \epsilon_t (p_t)) \inner{B u_t, v}} \leq \exp (\frac{1}{2} \norm{v}^2 + \sqrt{n} \gamma \norm{v} ) \label{eq:lemma_bound}
\end{align}
we want to split the sum into separate phases based on the used pivots $p$.
We regard the pivot phase using pivot $p\in[n]$ as starting in iteration $t_p$.
For $s\ge t_p-1$ we let $Q_{s,p} \subset[n]$ be the set of all positions $r$ decided in iteration $s$ during the pivot phase of $p$.
While $p$ is technically decided in iteration $t_p$, for analysis we count it in $Q_{t_p - 1, p}$.
Further, $Q_{s,p} = \emptyset$ for every $s$ after the iteration in which $p$ is frozen.

Let $\Delta_p$ be all the randomness up to the start of iteration $t_p$ point when $p$ is picked as the pivot (the choice of $p$ is included in $\Delta_p$).
We will analyze the contribution of pivot $p$'s phase to \eqref{eq:lemma_bound} by analyzing its contribution conditional on $\Delta_p$. In particular, given that conditioning, $p$, $t_p$, $Q_{t_p-1,p}$ are fixed. 
Technically, the end of the pivot $p$'s phase is random, but for analysis is would be easier if it were fixed.
To do this, we define $\delta_{t,p} = \delta_t, \epsilon_{t,p} = \epsilon_t$ and $=0$ for $t$ outside the phase. Thus we can take the sum $t=t_p...n$ for fixed end point $n$ such that
$$
\sum_{t=t_p}^n (\delta_{s,p} + \epsilon_{s,p} (p_t)) \inner{B u_t, v}
$$
is the contribution of pivot $p$'s phase to the sum in \eqref{eq:lemma_bound}.
In particular, a sum over the pivot phase is written $\sum_{t=t_p}^n$ and a sum over the positions decided in the phase is written $\sum_{s=t_p-1}^{n} \sum_{r \in Q_{s,p}}$. For the later we use that $Q_{s,p} = 0$ for $s$ outside the pivot's phase.

We reuse the notation $\alpha_r, \beta_r$ to denote $\inner{w_{\sigma(r)}, b_{p}}$ and $\inner{w_{\sigma(r)}, v}$, and we use $x$ to denote $z_t (p)$, the current position in $[-1,1]$ along the pivot unit when it improves notational clarity. Let $P_p$ be the projection matrix onto the $\mathrm{span}(\{w_{\sigma(r)}|r$ decided during pivot $p$'s phase$\})$

\begin{lemma}\label{lem:single_pivot}
For a pivot $p$ and $\gamma= (2 \epsilon_{pr} + \ec)n$ we have
$$\mathbb{E} \left[\exp \left( \sum_{t=t_p}^n (\delta_{t,p} + \epsilon_{t,p} (p)) \inner{B u_t, v} - \frac{1}{2} \norm{P_p b_p}^2 \norm{P_p v}^2 - \gamma \sum_{s=t_p-1}^{n} \sum_{r \in Q_{s,p}} \abs{\alpha_r \beta_r} \right) \middle| \Delta_p \right] \le 1$$
\end{lemma}

Before proving \Cref{lem:single_pivot}, we show that it implies \Cref{spielman:withinphasebound}.
In particular, \Cref{lem:single_pivot} can be seen as the contribution of pivot phase $p$ to \eqref{eq:lemma_bound}/\Cref{spielman:withinphasebound}.

\begin{proof}[Proof of \Cref{spielman:withinphasebound}]

\begin{align}
&\EEe{
\sum_{t=1}^T (\delta_t + \epsilon_t (p_t)) \inner{B u_t, v}} 
\exp\left(-\gamma \sqrt{n} \norm{v} - \frac{1}{2}\norm{v}^2\right) \notag\\
&\leq 
\mathbb{E} \left[ 
\prod_{p} \exp \left( 
    \sum_{t=t_p}^{n} (\delta_{t,p} + \epsilon_{t,p} (p)) \inner{B u_t, v} 
    - \frac{1}{2} \norm{P_p v}^2 
    - \gamma \sum_{s=t_p-1}^{n} \sum_{r \in Q_{s,p}} \abs{\alpha_r \beta_r} 
\right)  
\right] \notag\\
&\leq 
\mathbb{E} \left[ \prod_{p} \exp \left( \sum_{t=t_p}^{n} (\delta_{t,p} + \epsilon_{t,p} (p)) \inner{B u_t, v} - \frac{1}{2} \norm{P_p b_p}^2 \norm{P_p v}^2 - \gamma \sum_{s=t_p-1}^{n} \sum_{r \in Q_{s,p}} \abs{\alpha_r \beta_r} \right)  \right] \label{eq:3.2:boundby1}
\end{align}
where the first inequality holds pointwise because $\sum_{p} \norm{P_p v}^2 \leq \norm{v}^2$, as the $P_p$ project onto mutually orthogonal subspaces, and because by Cauchy--Schwarz
\begin{equation*}
    \sum_{p} \sum_{s=t_p-1}^{n} \sum_{r \in Q_{s,p}} \abs{\alpha_r \beta_r}
    \leq \Bigl(\sum_{p=1}^n \norm{P_p b_p}^2\Bigr)^{1/2} \Bigl(\sum_{p=1}^n \norm{P_p v}^2\Bigr)^{1/2}
    \leq \sqrt{n}\, \norm{v},
\end{equation*}
using orthonormality of the $\{w_{\sigma(r)}\}$ and $\norm{b_p} \leq 1$ which holds deterministically. 

Let 
$$X_p := \exp \left( \sum_{t=t_p}^{n} (\delta_{t,p} + \epsilon_{t,p} (p)) \inner{B u_t, v} - \frac{1}{2} \norm{P_p b_p}^2 \norm{P_p v}^2 - \gamma \sum_{s=t_p-1}^{n} \sum_{r \in Q_{s,p}} \abs{\alpha_r \beta_r} \right).$$
then by conditioning on $\Delta_p$ and \Cref{lem:single_pivot} we get for any $k \le n$ that $\E[X_k | \Delta_k] \le 1$ and thus
\begin{align*}
    \mathbb{E} \left[ \prod_{p=1}^k X_p  \right] &=
    \mathbb{E} \left[ X_{k} \cdot \prod_{p=1}^{k-1} X_p \right] \\
    &=
    \mathbb{E} \left[ \E[X_k | \Delta_k] \cdot \prod_{p=1}^{k-1} X_p \right] \\
    &\le
    \mathbb{E} \left[ \prod_{p=1}^{k-1} X_p \right].
\end{align*}
Inductively this implies $\mathbb{E} \left[ \prod_{p=1}^k X_p  \right] \le 1$. As in \cite{harshaw2023balancingcovariatesrandomizedexperiments} we note that the result holds for any fixed pivot ordering so we are free to take the above ordering for notation, with the full result following from a marginalization over all orderings.

At last, since \eqref{eq:3.2:boundby1} is bounded by 1, rearranging shows
\begin{equation*}
    \EEe{\sum_{t=1}^T (\delta_t + \epsilon_t (p_t)) \inner{B u_t, v}} \leq \exp \left(\frac{1}{2} \norm{v}^2 + \sqrt{n} \gamma \norm{v}\right)
\end{equation*}
as desired. \qedhere\\
\end{proof} 

\begin{proof}[Proof of \Cref{lem:single_pivot}]
By applying Lemma \ref{spielman:dirstepisproj} and linearity we have $\inner{B u_t, v} = \sum_{s=t_p-1}^{t-1} \sum_{r \in Q_{s,p}} \alpha_r \beta_r$ for pivot $p=p_t$, and that the set of $\{w_{\sigma(r)}\}$ where position $r$ is decided in pivot phase $p$ is an orthonormal basis for the subspace $P_p$ projects into, implying $\norm{P_p b_p}^2 = \sum_{s=t_p-1}^{n} \sum_{r \in Q_{s,p}} \alpha_r ^2$ and $\norm{P_p v}^2 = \sum_{s=t_p-1}^{n} \sum_{r \in Q_{s,p}} \beta_r ^2$. This yields
\begin{align*}
    \E\Bigg[\exp\Bigg(\sum_{t=t_p}^{n} (\delta_{s,p} + \epsilon_{s,p} (p)) \sum_{s=t_p-1}^{t-1} \sum_{r \in Q_{s,p}} \alpha_r \beta_r 
    -&~
    \frac{1}{2} \left(\sum_{s=t_p-1}^{n} \sum_{r \in Q_{s,p}} \alpha_r ^2\right) \left(\sum_{s=t_p-1}^{n} \sum_{r \in Q_{s,p}} \beta_r ^2\right)\\ 
    -&~ \gamma \sum_{s=t_p-1}^{n} \sum_{r \in Q_{s,p}} \abs{\alpha_r \beta_r}
    \Bigg) \Bigg| \Delta_{p}\Bigg].
\end{align*}
To prove \Cref{lem:single_pivot}, this expression must be bounded by 1.

For $R \in \{t_p - 1, \dots, n\}$ let us define
    \begin{align*}
    g(R) :=&~ \E\Bigg[\exp\Bigg(
    \sum_{t=t_p}^{n} (\delta_{t,p} + \epsilon_{t,p} (p)) \sum_{s=t_p-1}^{\min\{R,\,t-1\}} \sum_{r \in Q_{s,p}} \alpha_r \beta_r \\
    &~
    - \frac{1}{2} \left(\sum_{s=t_p-1}^{R} \sum_{r \in Q_{s,p}} \alpha_r ^2\right) \left(\sum_{s=t_p-1}^{R} \sum_{r \in Q_{s,p}} \beta_r ^2\right) \\
    &~
    - \gamma \sum_{s=t_p-1}^{R} \sum_{r \in Q_{s,p}} \abs{\alpha_r \beta_r}
    \Bigg)\Bigg|\Delta_p\Bigg].
\end{align*}

Observe that $g(n)$ recovers the same expression, since $\min\{n, t-1\} = t-1$ for every $t \leq n$.
We will prove that $g(n) \leq g(n-1) \leq \dots \leq g(t_p - 1) \leq 1$ which then implies \Cref{lem:single_pivot}.

Since $t_p$ is determined by $\Delta_{p}$, every summation range above is deterministic under this conditioning. We note that we're progressing not over the contributions to the MGF at each timestep in the pivot phase but rather over which decided positions $r$ we consider for the terms involving $\alpha_r, \beta_r$. We first point out that $g(n) \leq g(n-1)$, as the first summation
\begin{equation*}
\sum_{t=t_p}^{n} (\delta_{t,p} + \epsilon_{t,p} (p)) \sum_{s=t_p-1}^{\min\{R,\,t-1\}} \sum_{r \in Q_{s,p}} \alpha_r \beta_r
\end{equation*}
is identical for $R = n$ and $R = n-1$; every $t \leq n$ satisfies $\min\{n, t-1\} = \min\{n-1, t-1\} = t-1$. The other two summations in the above expectation only increase in absolute value as we sum over more terms, and so we have $g(n) \leq g(n-1)$. We now proceed to show for $R \in \{t_p, \dots, n-1\}$ we have $g(R) \leq g(R-1)$, with the final inequality $g(t_p - 1) \leq 1$ to be shown later via a similar technique.

\paragraph{Proving $g(R) \leq g(R-1)$:}

We claim it suffices to show that for any such $R$ we have

\begin{align*}
&\EE{\sum_{t=R+1}^{n} (\delta_{t,p} + \epsilon_{t,p} (p)) \eta_R - \frac{1}{2} a_R b_R - \gamma c_R}{\Delta_{R}} \\
&\qquad\qquad \leq \EE{\sum_{t=R+1}^{n} (\delta_{t,p} + \epsilon_{t,p} (p)) \eta_{R-1} - \frac{1}{2} a_{R-1} b_{R-1} - \gamma c_{R-1}}{\Delta_{R}}, 
\end{align*}
where we use $\Delta_R$ to denote randomness up to and including iteration $R$ (so it includes $\Delta_p$), and the shorthand $\eta_R := \sum_{s=t_p-1}^{R} \sum_{r \in Q_{s,p}} \alpha_r \beta_r$ with $a_R, b_R, c_R$ having the same summation but with the inner term being $\alpha_r^2, \beta_r^2$ or $\abs{\alpha_r \beta_r}$ respectively. The expectation on the left matches $g(R)$ except that the first summation considers only iterations $t \geq R+1$, for which $\min\{R, t-1\} = R$, and so all of the summations in the expectation involve the same terms.

To show the above inequality, we observe that the only randomness under this conditioning is in the realization of the $\delta_{t,p} + \epsilon_{t,p}(p)$. 

Note that this sum must take values in $\{1-x, -1 - x\}$, where $x := z_{R+1}(p)$ is the position of the pivot coordinate after iteration $R$: by our convention each coefficient $\delta_{t,p} + \epsilon_{t,p}(p)$ is the pivot increment $z_{t+1}(p) - z_t(p)$, so the sum telescopes to $\sum_{t=R+1}^{n} (\delta_{t,p} + \epsilon_{t,p} (p)) = z_{n+1} (p) - x$ with $z_{n+1} (p) \in \{1,-1\}$. In the case where $\sum_{\substack{t \in S_p \\ t > R}} \delta_{t,p} + \epsilon_{t,p} (p) $ is mean zero a calculation gives that the distribution takes value $1-x$ with probability $\frac{1+x}{2}$, however we don't assume that the noise terms are mean zero. We do however know that they are bounded and therefore have bounded mean conditioned on $\Delta_R$ and so we introduce $\tilde{x}$ to stand in for $x$, assuming that we take value $1-x$ with probability $\frac{1+\tilde{x}}{2}$ where $\abs{\tilde{x} - x} \leq \gamma$. 
Here $\abs{\tilde{x} - x} \leq \gamma$ holds for $\gamma = (2 \epsilon_{pr} + \ec) n$ by \Cref{appendix:gamma}.

We can therefore take expectations and rearrange to get
\begin{align*}
&~
\EEr{\sum_{t=R+1}^{n} (\delta_{t,p} + \epsilon_{t,p} (p)) \eta_R - \frac{1}{2} a_R b_R - \gamma c_R} \\
=&~
 \frac{1 + \tilde{x}}{2} \exp((1 - x) \eta_R - \frac{1}{2} a_R b_R - \gamma c_R) + \frac{1 - \tilde{x}}{2} \exp ((-1 - x)\eta_R - \frac{1}{2} a_R b_R - \gamma c_R) \\ 
=&~
 \exp(- \gamma c_R ) \left[ \frac{1 + \tilde{x}}{2} \exp((1 - x) \eta_R - \frac{1}{2} a_R b_R ) + \frac{1 - \tilde{x}}{2} \exp ((-1 - x)\eta_R - \frac{1}{2} a_R b_R )\right] \\ 
=&~
 \exp((\tilde{x} - x)\eta_R - \gamma c_R) \\
&\quad\quad\quad
\cdot\left[ \frac{1 + \tilde{x}}{2} \exp((1 - \tilde{x}) \eta_R - \frac{1}{2} a_R b_R ) + \frac{1 - \tilde{x}}{2} \exp ((-1 - \tilde{x})\eta_R - \frac{1}{2} a_R b_R )\right],
\end{align*}
where in the last step we used $1-x = (1-\tilde{x}) + (\tilde{x} - x)$ and $-1 - x = (-1 - \tilde{x}) + (\tilde{x} - x)$, factoring out $\exp ((\tilde{x} - x)\eta_R)$. This gives that \[\left[ \frac{1 + \tilde{x}}{2} \exp((1 - \tilde{x}) \eta_R - \frac{1}{2} a_R b_R ) + \frac{1 - \tilde{x}}{2} \exp ((-1 - \tilde{x})\eta_R - \frac{1}{2} a_R b_R )\right]=f_{\tilde{x}} \begin{bmatrix} a_R & \eta_R \\ \eta_R & b_R \end{bmatrix}.\]

We note that $\frac{1 + \tilde{x}}{2} \in [0,1]$ as it corresponds to the probability of a certain event and so $\tilde{x} \in [-1,1]$ so $f_{\tilde{x}}$ is defined. We can therefore apply lemma \ref{spielman:operatorfn} as for any $r$ the matrix $\begin{bmatrix} \alpha_r^2 & \alpha_r \beta_r \\ \alpha_r \beta_r & \beta_r^2 \end{bmatrix}$ is positive semi-definite and so $\begin{bmatrix} a_R & \eta_R \\ \eta_R & b_R \end{bmatrix}$ is a sum of positive semi-definite matrices.

We can therefore split off the final term of the summation common to $a_R, b_R, \eta_R$ and apply the fact $f$ is operator monotone decreasing to get
\begin{align*}
&~
\EEr{\sum_{t=R+1}^{n} (\delta_{t,p} + \epsilon_{t,p} (p)) \eta_R - \frac{1}{2} a_R b_R - \gamma c_R} \\
=&~
\exp((\tilde{x} - x) \eta_R - \gamma c_R) f_{\tilde{x}} \begin{bmatrix} a_R & \eta_R \\ \eta_R & b_R \end{bmatrix} \\
\leq&~
\exp((\tilde{x} - x) \eta_R - \gamma c_R) f_{\tilde{x}} \begin{bmatrix} a_{R-1} & \eta_{R-1} \\ \eta_{R-1} & b_{R-1} \end{bmatrix} \\
=&~
\exp((\tilde{x} - x)\eta_{R} - \gamma c_R) \\
&~
\cdot \left[ \frac{1 + \tilde{x}}{2} \exp((1 - \tilde{x}) \eta_{R-1} - \frac{1}{2} a_{R-1} b_{R-1} ) + \frac{1 - \tilde{x}}{2} \exp ((-1 - \tilde{x})\eta_{R-1} - \frac{1}{2} a_{R-1} b_{R-1} )\right],
\end{align*}
where we can now use a similar manipulation to factor out $\exp (x - \tilde{x})\eta_{R-1} = \exp (\tilde{x} - x)(-\eta_{R-1})$ to exchange the inner uses of $\tilde{x}$ with $x$, matching the expression for the original expectation.
\begin{align*}
=&~
\exp((\tilde{x} - x)(\eta_{R} - \eta_{R-1}) - \gamma c_R) \\
&~
\cdot\left[ \frac{1 + \tilde{x}}{2} \exp((1 - x) \eta_{R-1} - \frac{1}{2} a_{R-1} b_{R-1} ) + \frac{1 - \tilde{x}}{2} \exp ((-1 - x)\eta_{R-1} - \frac{1}{2} a_{R-1} b_{R-1} )\right] \\ 
\end{align*}
Note that we have $\abs{\tilde{x} - x} \leq \gamma$ by \Cref{appendix:gamma}, and so by the definitions of $\eta_R, c_R$ we have $(\tilde{x} - x) (\eta_R - \eta_{R-1}) - \gamma c_R \leq \gamma \sum_{r \in Q_R} \abs{\alpha_r \beta_r} - \gamma c_R = -\gamma c_{R-1} $. This gives
\begin{align*}
\leq&~ 
 \exp(- \gamma c_{R-1}) \Bigg[ \frac{1 + \tilde{x}}{2} \exp((1 - x) \eta_{R-1} - \frac{1}{2} a_{R-1} b_{R-1} ) \\
&~
 + \frac{1 - \tilde{x}}{2} \exp ((-1 - x)\eta_{R-1} - \frac{1}{2} a_{R-1} b_{R-1} )\Bigg] \\ 
=&~
 \Bigg[ \frac{1 + \tilde{x}}{2} \exp((1 - x) \eta_{R-1} - \frac{1}{2} a_{R-1} b_{R-1} - \gamma c_{R-1}) \\
&~ 
 + \frac{1 - \tilde{x}}{2} \exp ((-1 - x)\eta_{R-1} - \frac{1}{2} a_{R-1} b_{R-1} - \gamma c_{R-1})\Bigg] \\ 
=&~
\EEr{\sum_{t=R+1}^{n} (\delta_{t,p} + \epsilon_{t,p} (p)) \eta_{R-1} - \frac{1}{2} a_{R-1} b_{R-1} - \gamma c_{R-1}} 
\end{align*}
as desired. Having established this inequality we proceed to show $g(R) \leq g(R-1)$ as discussed. We will apply iterated expectations and condition on the first $R$ iterations of the algorithm, applying the lemma to the inner expectation, noticing that the resulting expression recovers $g(R-1)$ exactly. We have
\begin{align*}
    g(R) :=&~
    \mathbb{E} \Bigg[\exp \Bigg( \sum_{t=t_p}^{n} (\delta_{t,p} + \epsilon_{t,p} (p)) \sum_{s=t_p-1}^{\min\{R,\,t-1\}} \sum_{r \in Q_{s,p}} \alpha_r \beta_r \\
    &~
    - \frac{1}{2} \left(\sum_{s=t_p-1}^{R} \sum_{r \in Q_{s,p}} \alpha_r ^2\right) \left(\sum_{s=t_p-1}^{R} \sum_{r \in Q_{s,p}} \beta_r ^2\right) 
    - \gamma \sum_{s=t_p-1}^{R} \sum_{r \in Q_{s,p}} \abs{\alpha_r \beta_r} \Bigg) \,\Bigg|\, {\Delta_{p}} \Bigg] \\
    =&~
    \mathbb{E} \Bigg[ \mathbb{E} \Bigg[  \exp \Bigg( \sum_{t=t_p}^{n} (\delta_{t,p} + \epsilon_{t,p} (p)) \sum_{s=t_p-1}^{\min\{R,\,t-1\}} \sum_{r \in Q_{s,p}} \alpha_r \beta_r \\
    &~
    - \frac{1}{2} \left(\sum_{s=t_p-1}^{R} \sum_{r \in Q_{s,p}} \alpha_r ^2\right) \left(\sum_{s=t_p-1}^{R} \sum_{r \in Q_{s,p}} \beta_r ^2\right) - \gamma \sum_{s=t_p-1}^{R} \sum_{r \in Q_{s,p}} \abs{\alpha_r \beta_r} \Bigg) \,\Bigg|\, {\Delta_{R}}\Bigg] \,\Bigg|\, {\Delta_{p}} \Bigg] \\
    =&~
    \mathbb{E} \Bigg[ \exp \left( \sum_{t=t_p}^{R} (\delta_{t,p} + \epsilon_{t,p} (p)) \sum_{s=t_p-1}^{t-1} \sum_{r \in Q_{s,p}} \alpha_r \beta_r \right) \\
    &~
    \cdot \mathbb{E} \left[  \exp \left( \sum_{t=R+1}^{n} (\delta_{t,p} + \epsilon_{t,p} (p)) \eta_R - \frac{1}{2} a_R b_R - \gamma c_R \right) \,\middle|\, {\Delta_{R}}\right] \,\Bigg|\, {\Delta_{p}} \Bigg] \\
    \leq&~
    \mathbb{E} \Bigg[ \exp \left( \sum_{t=t_p}^{R} (\delta_{t,p} + \epsilon_{t,p} (p)) \sum_{s=t_p-1}^{t-1} \sum_{r \in Q_{s,p}} \alpha_r \beta_r \right) \\
    &~
    \cdot \mathbb{E} \left[  \exp \left( \sum_{t=R+1}^{n} (\delta_{t,p} + \epsilon_{t,p} (p)) \eta_{R-1} - \frac{1}{2} a_{R-1} b_{R-1} - \gamma c_{R-1} \right) \,\middle|\, {\Delta_{R}}\right] \,\Bigg|\, {\Delta_{p}} \Bigg] \\
\end{align*}
an check that in the outer summation ($t \leq R$) the index $s$ ranges up to $t-1 = \min\{R-1, t-1\}$, whereas in the inner summation ($t \geq R+1$) it ranges up to $R-1 = \min\{R-1, t-1\}$, allowing us to recombine as follows:
\begin{align*}
    =&~
    \mathbb{E} \Bigg[ \mathbb{E} \Bigg[  \exp \Bigg( \sum_{t=t_p}^{n} (\delta_{t,p} + \epsilon_{t,p} (p)) \sum_{s=t_p-1}^{\min\{R-1,\,t-1\}} \sum_{r \in Q_{s,p}} \alpha_r \beta_r \\
    &~
    - \frac{1}{2} \left(\sum_{s=t_p-1}^{R-1} \sum_{r \in Q_{s,p}} \alpha_r ^2\right) \left(\sum_{s=t_p-1}^{R-1} \sum_{r \in Q_{s,p}} \beta_r ^2\right) - \gamma \sum_{s=t_p-1}^{R-1} \sum_{r \in Q_{s,p}} \abs{\alpha_r \beta_r} \Bigg) \,\Bigg|\, {\Delta_{R}}\Bigg] \,\Bigg|\, {\Delta_{p}} \Bigg] \\
    =&~
    \mathbb{E} \Bigg[\exp \Bigg( \sum_{t=t_p}^{n} (\delta_{t,p} + \epsilon_{t,p} (p)) \sum_{s=t_p-1}^{\min\{R-1,\,t-1\}} \sum_{r \in Q_{s,p}} \alpha_r \beta_r \\
    &~
    - \frac{1}{2} \left(\sum_{s=t_p-1}^{R-1} \sum_{r \in Q_{s,p}} \alpha_r ^2\right) \left(\sum_{s=t_p-1}^{R-1} \sum_{r \in Q_{s,p}} \beta_r ^2\right) - \gamma \sum_{s=t_p-1}^{R-1} \sum_{r \in Q_{s,p}} \abs{\alpha_r \beta_r} \Bigg) \,\Bigg|\, {\Delta_{p}} \Bigg] \\
    =&~
    g(R-1)
\end{align*}
as claimed. We remark that the case $g(t_p - 1) \leq 1$ can be shown with the same argument as above, without needing to apply iterated expectations as this iteration only counts the pivot $b_p$ as decided. In particular, with $x = z_{t_p}(p)$ we have
\begin{align*}
g(t_p - 1) &= \EE{\sum_{t=t_p}^{n} (\delta_{t,p} + \epsilon_{t,p} (p)) \eta_{t_p-1} - \frac{1}{2} a_{t_p-1} b_{t_p-1} - \gamma c_{t_p-1}}{\Delta_{p}} \\
    &= \exp( (\tilde{x} - x) \eta_{t_p-1} - \gamma c_{t_p-1})  f_{\tilde{x}} \begin{bmatrix} a_{t_p-1} & \eta_{t_p-1} \\ \eta_{t_p-1} & b_{t_p-1} \end{bmatrix} \\
    &\leq 1 \cdot f_{\tilde{x}} \begin{bmatrix} a_{t_p-1} & \eta_{t_p-1} \\ \eta_{t_p-1} & b_{t_p-1} \end{bmatrix} \\
    &\leq f_{\tilde{x}} \begin{bmatrix} 0 & 0 \\ 0 & 0 \end{bmatrix} = 1
\end{align*}
which completes the proof.
\end{proof}

\subsection{Proof of Theorem 2.3}
\label{sec: pf of theorem 2.3}

\begin{proof}[Proof of Theorem~\ref{thm:spielmanadapted}]

As discussed in \eqref{eq:overview:split} we can rewrite our output $Bz - Bz_0$ as
\begin{align*}
Bz - Bz_0 &= B \left( \sum_{t=1}^T \dth u_t + \dth(\uth - u_t) + \clamp \right)  \\
&= B \sum_{t=1}^T \left(\dth + \clamp(p_t)\right) u_t + B \left(\sum_{t=1}^T \dth (\uth - u_t)+ \clamp - \clamp(p_t) u_t\right) 
\end{align*}
where we are able to use a modified version of the proof in \cite{harshaw2023balancingcovariatesrandomizedexperiments} to bound the MGF of the first term. We use $\ec$ as an upper bound on the infinity norm of $\clamp$, deferring the choice of a concrete constant. Considering the remaining terms observe that we have
\begin{gather*}
\norm{B(\dth (\uth - u_t) + \clamp - \clamp(p_t) u_t)}_2 \leq \abs{\dth} \norm{B (\uth - u_t)}_2 + \norm{B \clamp}_2 + |\clamp (p_t)| \norm{B u_t}_2 \\
\end{gather*}
and from here we note that $\norm{B u_t} \leq 1$ by the fact $B u_t$ is some projection of the pivot unit, and $\norm{B \clamp} = \norm{\sum_{i} \clamp (i) b_i} \leq \sum_i \abs{\clamp (i)} \norm{b_i} \leq \sum_i \abs{\clamp (i)} \leq n \ec$
which together with the fact $\abs{\dth} \leq 2$ gives
\begin{equation*}
    \sum_{t=1}^n \norm{B(\dth (\uth - u_t) + \clamp - \clamp(p_t) u_t)}_2 \leq \sum_{t=1}^n (2\els + n \ec + \ec) \leq 2n  \els + n(n+1)\ec
\end{equation*}

which implies that $\exp(\sum_t \inner{B e_t, v}) \leq \exp ( (2 n \els  +n(n+1) \ec) \norm{v})$ by the Cauchy-Schwarz inequality where $e_t$ is the split of error term discussed in \ref{paragraph:error_decomp}. This allows us to show
\begin{align*}
    \E[\exp (\inner{Bz - Bz_0),v}] &= \E \left[\exp \left(\sum_t (\delta_t + \epsilon_t (p_t))\inner{B u_t, v} + \inner{B e_t, v} \right)\right]  \\
    &\leq \exp \left(\frac{1}{2} \norm{v}^2 + (\sqrt{n} \gamma + (2 n \els  +n(n+1) \ec)) \norm{v}\right)
\end{align*}
by applying the above bound on the $e_t$ terms, followed by lemma \ref{spielman:withinphasebound}. As will be shown in \cref{appendix:constants}, we can choose $\gamma = (2 \epsilon_{pr} + \ec) n$ to get a bound of
\begin{equation*}
    \exp \left(\frac{1}{2} \norm{v}^2 + \left(\underbrace{n^{3/2}(2 \epsilon_{pr} + \ec)}_{\text{drift/sampling}} + \left(\underbrace{2 n \cdot \els}_{\text{least squares}}  +\underbrace{n(n+1) \cdot \ec}_{\text{arithmetic/clamping}}\right)\right)\norm{v}\right).
\end{equation*}
\end{proof}

\subsection{Choosing Constants} \label{appendix:constants}
We now discuss how to choose the constant $\gamma$ used in our result and how to further decompose $\epsilon_{pr}, \els$ to give a more precise characterization of error in terms of machine dependent quantities. 

\begin{lemma} \label{appendix:gamma}
    For each iteration $R$ within a pivot phase for pivot $p$, defining $\tilde{x}$ such that $\Pr[\sum_{R+1}^n (\delta_{t,p} + \epsilon_{t,p} (p)) =1-x \,\big|\, \Delta_R] = \frac{1 + \tilde{x}}{2}$ we have $\abs{\tilde{x} - x} \leq (2 \epsilon_{pr} + \ec)n$ where $x := z_{R+1} (p)$, and in the ideal algorithm the corresponding probability is $\frac{1+x}{2}$.
\end{lemma}
\begin{proof}

Focusing first on $\gamma$, we previously stated that $\gamma$ is our upper bound on $\abs{\tilde{x} - x}$ where $\tilde{x}$ is defined such that the probability that $\sum_{t=R+1}^{n} (\delta_{t,p} + \epsilon_{t,p} (p))$ takes value $1 - x$, conditional on $\Delta_R$, is $\frac{1 + \tilde{x}}{2}$. We now motivate this definition further and elaborate on how we arrive at our concrete choice of $\gamma$.

Recall that, conditioned on $\Delta_R$, $\sum_{t=R+1}^{n} (\delta_{t,p} + \epsilon_{t,p} (p))$ obeys a two point distribution, taking values in $\{1-x, -1 - x\}$. Therefore if we use $\mu$ to denote its conditional mean, we can write for some $q \in [0,1]$ that
\begin{align*}
    \mu &= q(1-x) + (1-q)(-1-x) \\
    &= 2 q - 1 - x
\end{align*}
and therefore $q = \frac{1+x+\mu}{2}$. Setting $q = \frac{1 + \tilde{x}}{2}$ gives that $\tilde{x} = x + \mu$ and so our assumption of $\abs{\tilde{x} - x} \leq \gamma$ is therefore equivalent to assuming that $\abs{\mu} \leq \gamma$. Put another way, $\tilde{x}$ is simply a notational alternative to writing $x + \E [\sum_{t=R+1}^{n} (\delta_{t,p} + \epsilon_{t,p} (p)) \mid \Delta_R]$ throughout the proof, emphasizing that it is the result of perturbing $x$ (which would be the value in the infinite precision case) slightly due to noise.

We now discuss the value of $\gamma$. To provide a value for $\gamma$ it suffices to bound
\begin{equation*}
\left | \E \left[\sum_{t=R+1}^{n} (\delta_{t,p} + \epsilon_{t,p} (p)) \,\Big|\, \Delta_R\right] \right | = 
\left| \sum_{t=R+1}^{n} \E \bigl[\delta_{t,p} + \epsilon_{t,p} (p) \,\big|\, \Delta_R\bigr] \right |
\end{equation*}

    Let $\mathcal F_t$ contain the history and all computations preceding
 the random branch choice at iteration $t$, and let $\mathcal H_t$
 contain the history through the completed update at iteration $t$.
Thus $\Delta_R=\mathcal H_R$ and
$\mathcal H_R\subseteq\mathcal F_t$ for $t>R$.
Write $I_{t,p}$ for the indicator that iteration $t$ belongs to the
phase of pivot $p$, noting that this indicator is $\mathcal F_t$-measurable. Before termination,
$\widehat p_t$ differs from the probability centering the computed
endpoints by at most $\varepsilon_{\rm pr}$, so
\[
 \left|\mathbb E[\widehat\delta_t\mid\mathcal F_t]\right|
 = (\widehat\delta_t^++\widehat\delta_t^-)
   |\widehat p_t-p_t^{\rm comp}|
 \le 2\varepsilon_{\rm pr}.
\]
After termination the padded step has zero conditional mean.
Moreover, $|\varepsilon_t(p)|\le\varepsilon_{\rm iter}$, including
when the update error depends on the sampled branch. Therefore,
\[
 \left|\mathbb E\!\left[
 I_{t,p}\bigl(\widehat\delta_t+\varepsilon_t(p)\bigr)
 \mid\mathcal F_t\right]\right|
 \le I_{t,p}(2\varepsilon_{\rm pr}+\varepsilon_{\rm iter}).
\]
The tower property and the triangle inequality now give
\begin{align*}
 |\mu|
 &=\left|\mathbb E\!\left[
 \sum_{t=R+1}^{n}I_{t,p}
 (\widehat\delta_t+\varepsilon_t(p))\mid\mathcal H_R
 \right]\right|\\
 &\le\sum_{t=R+1}^{n}
 \mathbb E\!\left[
 \left|\mathbb E\!\left[
 I_{t,p}(\widehat\delta_t+\varepsilon_t(p))\mid\mathcal F_t
 \right]\right|\mid\mathcal H_R\right]\\
 &\le(n-R)(2\varepsilon_{\rm pr}+\varepsilon_{\rm iter})
 \le\gamma
\end{align*}
Since $\widetilde x=x+\mu$, this proves
$|\widetilde x-x|\le\gamma$ for $\gamma:=n(2\varepsilon_{\rm pr}+\varepsilon_{\rm iter})$.

\end{proof}

\begin{lemma}[Bounding sampling error, arithmetic error, in terms of machine precision] \label{consants:prob_error}
On a floating point machine with $b$ random bits, we have $\epsilon_{pr} \leq \frac{2u}{1 - 2u} + 2^{-b}$ for unit roundoff $u$.
\end{lemma}
\begin{proof}
    
We note that in the following discussions $u$ refers to unit roundoff while $u_t$ is the step direction at time $t$. In the following discussion we condition on the history and all computations preceding the random branch choice at iteration $t$. Let $p_t^{\mathrm{comp}}\coloneqq \frac{\hat\delta_t^-}{\hat\delta_t^-+\hat\delta_t^+}$, where $\hat\delta_t^+\geq 0$, $\hat\delta_t^-\geq 0$, and $\hat\delta_t^++\hat\delta_t^->0$. Here, $p_t^{\mathrm{comp}}$ is the exact ratio of the stored endpoint magnitudes, and it hence centers the step at the computed endpoints.

We assume binary floating-point arithmetic with $L$ significand bits (including the leading bit), rounding to nearest, and no overflow or underflow in the operations below. Here the unit roundoff is $u=2^{-L}$. Compute $s_t=\operatorname{fl}(\hat\delta_t^-+\hat\delta_t^+)$, $r_t=\operatorname{fl}(\hat\delta_t^-/s_t)$, and $q_t=\min\{1,\max\{0,r_t\}\}$.

Then, by the standard floating-point model (in Higham's lemma, \citep{doi:10.1137/1.9780898718027}), there exists $\theta_{1,t},\theta_{2,t}$ with $|\theta_{i,t}|\le u$ such that $r_t=p_t^{\mathrm{comp}}\frac{1+\theta_{2,t}}{1+\theta_{1,t}}$. Since clipping to $[0,1]$ cannot increase the error and $p_t^{\mathrm{comp}}\in[0,1]$, for $2u<1$ we have
\begin{align*}
|q_t-p_t^{\mathrm{comp}}|
&\le |r_t-p_t^{\mathrm{comp}}|\\
&=p_t^{\mathrm{comp}}
  \frac{|\theta_{2,t}-\theta_{1,t}|}{|1+\theta_{1,t}|}\\
&\le \frac{2u}{1-u}\\
&\le \frac{2u}{1-2u}.
\end{align*}
To sample using $b$ fresh independent unbiased random bits, form $k_t=\lfloor 2^b q_t\rfloor$ exactly, draw $J_t$ uniformly from $\{0,\ldots,2^b-1\}$ independently of the conditioned history, and choose the positive branch if $J_t<k_t$. The resulting conditional positive-branch probability is $\hat p_t=k_t/2^b$, so we have that $|\hat p_t-q_t|<2^{-b}$. Then, the triangle inequality therefore gives, uniformly over all
iterations and histories, that \[|\hat p_t-p_t^{\mathrm{comp}}|
\le \frac{2u}{1-2u}+2^{-b}.\]
Therefore, we can set $\epsilon_{\mathrm{pr}}
=\frac{2^{1-L}}{1-2^{1-L}}+2^{-b}$.

\end{proof}

\begin{lemma}[Least-squares image error from backward stability]
\label{corollary: ls-backward-error}
Suppose that $B\in\mathbb R^{m\times n}$ has full column rank and $\|b_i\|_2\le1$ for every $i$. Write $\kappa_B:=\|B\|_2/\sigma_{\min}(B)$. At iteration $t$, let $S_t=A_t\setminus\{p_t\}$, $M_t=B_{S_t}$ and $b=b_{p_t}$. The exact nonpivot coefficients solve $x_t\in\arg\min_x\|M_tx+b\|_2$. Suppose that the computed coefficients $\hat x_t$ satisfy the following normwise backward-error guarantee, uniformly over all iterations and execution histories: for some perturbations $E_t,f_t$, \[ \hat x_t\in\arg\min_x \|(M_t+E_t)x+b+f_t\|_2,\] where $\|E_t\|_2\le\eta\|M_t\|_2$ and $\|f_t\|_2\le\eta\|b\|_2$, where $\eta\ge0$ is deterministic. Form $u_t,\hat u_t$ by placing $x_t,\hat x_t$ on $S_t$, setting the pivot entry to one, and setting all inactive entries to zero. If $\eta\kappa_B<1$, then we have that $\|B(\hat u_t-u_t)\|_2\le\epsilon_{\mathrm{LS}}$ for every $t$, where \[ \epsilon_{\mathrm{LS}} \leq \frac{3\eta\kappa_B}{1-\eta\kappa_B}.\]  In particular, if $\eta\kappa_B\le\tfrac12$, one may instead take $\epsilon_{\mathrm{LS}}=6\eta\kappa_B$.
\end{lemma}

\begin{proof}
    We clarify that the assumptions of full column rank on $B$ and backwards stability in solvers are not true in general and that this result is meant to provide an example for what $\els$ scales like in certain scenarios. In particular the Higgs matrix in the experiments is an example of a wide matrix that lacks full column rank.
    When $S_t=\emptyset$, both directions equal the pivot unit vector and the image error is zero. Otherwise, suppress the iteration index and write $M=M_t$, $\widetilde M=M+E$, and $q=\eta\kappa_2(M)$. Since $M$ consists of columns of $B$, we have that $\sigma_{\min}(M)\ge\sigma_{\min}(B)$, with $\|M\|_2\le\|B\|_2$ and $\kappa_2(M)\le\kappa_B$. Thus $q<1$, and \begin{align*}\sigma_{\min}(\widetilde M) &\ge\sigma_{\min}(M)-\|E\|_2 \\
    &\ge(1-q)\sigma_{\min}(M) >0.\end{align*}

Let $P$ and $\widetilde P$ be the orthogonal projectors onto the column spaces of $M$ and $\widetilde M$. These spaces have equal dimension, so the equal-rank projector identity gives 
\begin{align*} \|P-\widetilde P\|_2 &=\|(I-\widetilde P)P\|_2\\ &=\|(I-\widetilde P)E M^\dagger\|_2 \\
&\le\frac{\|E\|_2}{\sigma_{\min}(M)} \le q, \end{align*}
where $M^\dagger$ denotes the Moore-Penrose pseudoinverse and $(I-\widetilde P)M=-(I-\widetilde P)E$. Then, least-squares optimality gives $Mx_t=-Pb$, $\widetilde M\hat x_t=-\widetilde P(b+f)$, and $\|\hat x_t\|_2 \le\frac{(1+\eta)\|b\|_2} {(1-q)\sigma_{\min}(M)}$.

Hence, we have that
\begin{align*} \|B(\hat u_t-u_t)\|_2 &=\|M(\hat x_t-x_t)\|_2\\ 
&= \|(P-\widetilde P)b-\widetilde P f-E\hat x_t\|_2 \\
&\le\left(q+\eta+\frac{q(1+\eta)}{1-q}\right)\|b\|_2\\ 
&=\frac{\eta+2q-q^2}{1-q}\|b\|_2 \\
&\le\frac{3q}{1-q} \\
&\le\frac{3\eta\kappa_B}{1-\eta\kappa_B}, \end{align*}

where we used $\|b\|_2\le1$, $\kappa_2(M)\ge1$, and $\eta\le q$, which completes the proof. The simplified bound comes from substituting $\eta\kappa_B\le\tfrac12$.\qedhere\\
\end{proof}

\begin{remark}For a solver whose backward-error guarantee is $\eta\le C_{\mathrm{LS}}(m,n) u$, with the dimension-dependent factor chosen uniformly over all active systems, the sufficient precision condition $C_{\mathrm{LS}}(m,n)u \kappa_B \le\tfrac12$ therefore permits $\epsilon_{\mathrm{LS}} =6C_{\mathrm{LS}}(m,n)\kappa_Bu$. Its contribution to the linear term in Theorem \ref{thm:spielmanadapted} is bounded by \[ 2n\epsilon_{\mathrm{LS}} \le12nC_{\mathrm{LS}}(m,n)\kappa_B u. \]
\end{remark}
In particular the discussion in \ref{consants:prob_error}, \ref{corollary: ls-backward-error} and the relevant discussions imply a bound of
    \begin{align*}
        \E [\exp(\inner{B(z - z_0), v})]&\leq
        \exp\bigg(\frac{1}{2} \norm{v}^2+ \bigg(n^{3/2}\bigg(2\bigg(\frac{2u}{1 - 2u} + 2^{-b}\bigg) + \ec \bigg)\\
        &\qquad\qquad\qquad+\bigg(12nC_{\mathrm{LS}}(m,n)\kappa_B u +n(n+1) \ec \bigg)\bigg)\norm{v}\bigg).
    \end{align*}  
Finally, we prescribe bit prescriptions in order to guarantee an error of at most $\tau$ in the Gram-Schmidt walk in Corollary \ref{cor:sufficient_bits}. For this, we introduce the following helpful lemma below.

\begin{lemma}[Arithmetic and freezing error]
\label{lem:update_precision}
Let $\mathfrak u=2^{-L}\le1/4$ be the unit roundoff
of binary arithmetic with $L$ significand bits,
including the leading bit. Assume no overflow or
underflow in the operations below. Write
$\operatorname{RN}$ for rounding to nearest and
$\operatorname{RD}$ for rounding toward $-\infty$. Then, at a stored state $z_t\in[-1,1]^n$, let
$\widehat u_t$ be the stored direction, with pivot
entry one and zero entries outside the active set.
For each branch $s\in\{-1,+1\}$ and active coordinate
with $\widehat u_t(i)\ne0$, compute
\[a_{t,i}^s
 =\operatorname{sgn}(s\widehat u_t(i)),
\qquad
q_{t,i}^s
 =\operatorname{RD}\left(
 \frac{\operatorname{RD}(1-a_{t,i}^s z_t(i))}
 {|\widehat u_t(i)|}
 \right).\]
Set $d_t^s=\min_i q_{t,i}^s$ and $i_t^s\in\arg\min_i q_{t,i}^s$, where both endpoints are computed before choosing the branch.
After choosing branch $s$, set
$\widehat\delta_t=s d_t^s$ and compute, coordinatewise,
\[
w_t(i)
 =\operatorname{RN}\left(
 z_t(i)+
 \operatorname{RN}(\widehat\delta_t\widehat u_t(i))
 \right).
\]
Clip these values to $[-1,1]$, and set
$z_{t+1}(i_t^s)=a_{t,i_t^s}^s$.
Any additional snapping to an endpoint is permitted
only when it changes the clipped value by at most
$c_{\mathrm{snap}}\mathfrak u$. Here, previously frozen coordinates are left unchanged, and additional snapping is applied only to active coordinates. Then the update is feasible, freezes at least one additional coordinate, and satisfies $d_t^+\le1-z_t(p_t)$ and $d_t^-\le1+z_t(p_t)$. Moreover, $z_{t+1} =z_t+\widehat\delta_t\widehat u_t+\varepsilon_t$ with $\|\varepsilon_t\|_\infty
 \le C_{\mathrm{it}}\mathfrak u$ and 
$C_{\mathrm{it}}:=8+c_{\mathrm{snap}}$.
\end{lemma}

\begin{proof}
Fix a branch and define the exact boundary distances
\begin{equation}\rho_i
 =\frac{1-a_{t,i}^s z_t(i)}
 {|\widehat u_t(i)|}.
\end{equation}
Each downward rounding of a positive normal result
loses at most a relative $2\mathfrak u$, so
$(1-2\mathfrak u)^2\rho_i
 \le q_{t,i}^s\le\rho_i$.
Therefore $d_t^s\le\min_i\rho_i$, and
$y=z_t+s d_t^s\widehat u_t$ is feasible.
The pivot entry gives the stated endpoint caps.
At the selected binding coordinate, the distance
between $y$ and the endpoint used for explicit
freezing is at most
$\left[1-(1-2\mathfrak u)^2\right]
 (1-a_{t,i_t^s}^s z_t(i_t^s))
 \le8\mathfrak u$. 
 
 For every other updated coordinate, feasibility gives
$|\widehat\delta_t\widehat u_t(i)|\le2$ and
$|y(i)|\le1$. The rounded multiplication and addition
therefore incur total absolute error at most
$2\mathfrak u+\mathfrak u(1+2\mathfrak u)
 \le4\mathfrak u$.
Clipping cannot increase the distance to $y(i)$.
Additional snapping contributes at most
$c_{\mathrm{snap}}\mathfrak u$.
Combining these bounds proves the claim.\qedhere\\
\end{proof}

\begin{corollary}[Sufficient significand precision]
\label{cor:sufficient_bits}
Assume $\|b_i\|_2\le1$, that $z_0$ is exactly
representable at the working precision, and use
the implementation of 
Lemma~\ref{lem:update_precision}, retaining each pivot
until it freezes. Suppose the least-squares solver
satisfies, uniformly over all iterations and execution
histories, $\varepsilon_{\mathrm{LS}}
\le6C_{\mathrm{LS}}\kappa_\star\mathfrak u$ and $C_{\mathrm{LS}}\kappa_\star\mathfrak u\le\tfrac12$, where $C_{\mathrm{LS}}$ is a deterministic constant
from the solver's backward-error bound, and
$\kappa_\star$ is a uniform conditioning bound for
the active systems. Compute and sample the branch
probability as in Lemma~\ref{consants:prob_error}, using $r$ fresh unbiased
random bits at each iteration. Let\looseness=-1
\begin{equation}
\Gamma
 =(8+C_{\mathrm{it}})n^{3/2}
 +12nC_{\mathrm{LS}}\kappa_\star
 +C_{\mathrm{it}}n(n+1).
\end{equation}
For any target additive penalty $\tau>0$, it suffices
to choose
\[L\ge
\left\lceil
\log_2\max\left\{
4,\,
2C_{\mathrm{LS}}\kappa_\star,\,
\frac{2\Gamma}{\tau}
\right\}
\right\rceil
\]
and
\[r\ge
\left\lceil
\log_2\max\left\{
2,\,
\frac{4n^{3/2}}{\tau}
\right\}
\right\rceil.\]
Then the linear MGF penalty in Theorem~\ref{thm:spielmanadapted} satisfies
$A\le\tau$. In particular, for every $\alpha\in(0,1)$, we have
\[\Pr\left[
\|B(z-z_0)\|_\infty >\sqrt{2\log(2m/\alpha)}+\tau
\right]\le\alpha.\]
\end{corollary}

\begin{proof}
Since $\mathfrak u\le1/4$, Lemma~\ref{consants:prob_error} gives
\[\varepsilon_{\mathrm{pr}}
 \le\frac{2\mathfrak u}{1-2\mathfrak u}+2^{-r}
 \le4\mathfrak u+2^{-r}.\]
Substituting this bound, the assumed least-squares
bound, and
$\varepsilon_{\mathrm{iter}}\le C_{\mathrm{it}}\mathfrak u$
into Theorem~\ref{thm:spielmanadapted} yields
\[
A
 \le\Gamma\,2^{-L}+2n^{3/2}2^{-r}.
\]
The prescribed choices bound each term by $\tau/2$.
The discrepancy guarantee follows by applying the
MGF bound to each signed coordinate direction,
optimizing the Chernoff parameter, and taking a
union bound, which completes the proof.\qedhere\\
\end{proof}

\begin{remark}For fixed $c_{\mathrm{snap}}$ and $0<\tau\le1$, these
sufficient requirements have the orders
\[L=O\left(\log\frac{n^2+nC_{\mathrm{LS}}\kappa_\star}{\tau}
\right),
\qquad
r=O\left(
 \log\frac{n^{3/2}}{\tau}
\right),\]
where $L$ counts significand bits including the leading
bit.\end{remark}

\section{Lower Bound}
\label{sec: lower bound}

\begin{theorem}[Lower bound for bounded update errors]
There exists a universal constant $c>0$ such that, for every integer $n\ge 8$ and every $\varepsilon\ge 0$, there is a full-column-rank matrix $\mB=[b_1,\ldots,b_n]\in\mathbb{R}^{(n+1)\times n}$ with $\|b_i\|_2=1$ and an admissible error process for the Gram-Schmidt Walk, initialized at $z_0=0$, such that $z_{t+1}=z_t+\delta_tu_t+\epsilon_t$ with $\|\epsilon_t\|_\infty\le\varepsilon$, with exact step-selection probabilities, and whose terminal discrepancy $D=\mB z_T$ satisfies \[  \| \mathbb{E}D\|_2 \ge c\min\{n^2\varepsilon,n\}, \qquad\|D\|_{\psi_2} \ge c\min\{n^2\varepsilon,n\}.\] Here $\|\cdot\|_{\psi_2}$ denotes the uncentered directional Orlicz norm. The errors are deterministic functions of the
iteration, active set, and current pivot, which preserve all
frozen coordinates, and the walk remains feasible and
terminates after at most $n$ iterations.
\end{theorem}

\begin{proof}
Let $e_0,e_1,\ldots,e_n$ be the standard basis of
$\mathbb{R}^{n+1}$, and for all $i\in[n]$, set $b_i=\frac{e_0+e_i}{\sqrt{2}}$. We use a zero index based notation here for convenience; this convention is local to the lower bound.
Each column has Euclidean norm one, and
\[\mB^\top \mB
    =\frac12\bigl(\mI_n+\mathbf{1}\mathbf{1}^\top\bigr)\]
is positive definite, so $\mB$ has full column rank.

We first identify the exact update direction. For an active set $A$ of size $k$ and pivot $p\in A$, every feasible direction satisfies $u(p)=1$ and $u(j)=0$ for $j\notin A$. Moreover,
\[2\|\mB u\|_2^2 = \left(\sum_{j\in A}u(j)\right)^2+\sum_{j\in A}u(j)^2.\]
Minimizing this strictly convex quadratic gives $u(j)=-\frac1k$ for $j\in A\setminus \{p\}$ and $\sum_{j\in A}u(j) = \frac1k$.
In particular, we have $\langle Bu,e_0\rangle=\frac{1}{\sqrt{2}\,k}$. Set
 $L=\left\lfloor\frac n4\right\rfloor$ with $a=\min\left\{\varepsilon,\frac{1}{8n}\right\}$. Writing $A_t$ and $p_t$ for the active set and pivot before
iteration $t$, define
\[\epsilon_t(j)=
    \begin{cases}
        a,& 0\le t<L,\quad j\in A_t\setminus\{p_t\},\\
        0,& \text{otherwise}.
\end{cases}\]
Thus $\|\epsilon_t\|_\infty\le\varepsilon$, and the errors vanish on the pivot and all previously frozen coordinates.

We now verify via mathematical induction that during the first $L$ iterations, both feasible endpoints are determined by the pivot and exactly that coordinate freezes. We argue by induction. Suppose this holds before iteration $t<L$, so that $|A_t|=n-t$. The two candidate pivot endpoints correspond to signed steps $\delta=1-z_t(p_t)$ and $\delta=-1-z_t(p_t)$, each of magnitude at most two.

Any $j\in A_t\setminus\{p_t\}$ has been an active non-pivot
coordinate in every preceding iteration: a previous pivot
would already have frozen. Hence, for either candidate
step, its ideal updated value and its perturbed updated
value both have absolute value at most
\[
\begin{aligned}
    \sum_{s=0}^{t}
        \left(\frac{2}{n-s}+a\right)
    &\le
    \frac{2L}{n-L+1}+La\\
    &\le
    \frac23+\frac1{32}
    =\frac{67}{96}<1.
\end{aligned}
\]

Previously frozen coordinates remain unchanged, while the pivot reaches $1$ or $-1$ and receives zero error. Therefore both candidate endpoints are feasible before and after adding the error. Since extending either candidate step would move the pivot outside the cube, these are precisely the exact feasible endpoints. Every other active coordinate remains strictly interior, so exactly the pivot freezes. This completes the induction. After iteration $L-1$, all errors vanish and the algorithm continues as an exact Gram-Schmidt Walk. Thus the entire process is admissible and terminates at some $T\le n$. Now, it remains to bound the accumulated discrepancy. Let
\[Y=\langle Bz_T,e_0\rangle.\]
We condition on the history before selecting the step, including the current pivot. The exact step-selection probabilities ensure that
$\mathbb{E}[\delta_t\mid\mathcal{F}_t]=0$.
Since $u_t$ is determined by this history, each ideal
discrepancy increment has conditional mean zero.
Padding terminated runs with zero updates therefore gives
\begin{align*}\mathbb{E}Y&=\mathbb{E}\sum_{t=0}^{T-1}
        \langle B\epsilon_t,e_0\rangle\\
        &=\frac{a}{\sqrt{2}}
        \sum_{t=0}^{L-1}(n-t-1).
\end{align*}
For $n\ge8$, we have $L\ge n/8$ and $n-L\ge3n/4$.
Consequently, \begin{align*}\mathbb{E}Y&\ge\frac{a}{\sqrt{2}}L(n-L)\\
&\ge\frac{3}{32\sqrt{2}}an^2.\end{align*}
Furthermore, $an^2=\min\left\{n^2\varepsilon,\frac n8\right\}\ge\frac18\min\{n^2\varepsilon,n\}$, which proves the lower bound on
$\|\mathbb{E}D\|_2$ as $e_0$ is a unit vector. Finally, for every $K>0$, Jensen's inequality gives
\[\mathbb{E}\exp(Y^2/K^2)\ge\exp\bigl((\mathbb{E}Y)^2/K^2\bigr).\]
Using
$\|Y\|_{\psi_2}=\inf\{K>0:\mathbb{E}\exp(Y^2/K^2)\le2\}$,
it follows that
\begin{align*}\|D\|_{\psi_2}&\ge \|Y\|_{\psi_2}\\
&\ge \frac{|\mathbb{E}Y|}{\sqrt{\log 2}}\\
&=\Omega\bigl(\min\{n^2\varepsilon,n\}\bigr),\end{align*}
which completes the proof.\qedhere\\
\end{proof}

\section{Experimental Details}
\label{sec: experimental details}

For the distributional experiments in Figures~\ref{fig:bits}--\ref{fig:noise}, we generate $R=1000$ independent walks per configuration, parallelized across CPU processes on a six-core laptop with single-threaded BLAS.  Let $\mathcal D$
denote the finite set of tested unit directions, and write
$X_r(d)=d^\top B_0z^{(r)}$, where $B_0$ is the original
input matrix. We split the walks into disjoint selection
and evaluation sets, $\mathcal I_{\mathrm{sel}}$ and
$\mathcal I_{\mathrm{ev}}$, of equal size.
For either set $\mathcal I$, define
$\overline X_{\mathcal I}(d)
 =\frac{1}{|\mathcal I|}
   \sum_{r\in\mathcal I}X_r(d)$, and the centered moment diagnostic
\begin{equation}
\widehat\sigma_{\mathcal I}(d)
 =\max_{j\in\{1,2,3\}}
 \left(
 \frac{
 |\mathcal I|^{-1}\sum_{r\in\mathcal I}
 |X_r(d)-\overline X_{\mathcal I}(d)|^{2j}
 }{(2j-1)!!}
 \right)^{\frac{1}{2j}}.
 \end{equation}
We select $d_\star\in\arg\max_{d\in\mathcal D}
\widehat\sigma_{\mathcal I_{\mathrm{sel}}}(d)$ and report
$\widehat\sigma_{\mathcal I_{\mathrm{ev}}}(d_\star)$.
The evaluation walks are independent of those used to
select the direction. In particular, $d_\star$ need not
maximize the evaluation-set diagnostic.

On the evaluation set, we also compute the empirical
Orlicz statistic
\begin{equation}
\widehat K(d)
 =\inf\left\{c>0:
 \frac{1}{|\mathcal I_{\mathrm{ev}}|}
 \sum_{r\in\mathcal I_{\mathrm{ev}}}
 \exp\left(
 \frac{(X_r(d)-\overline X_{\mathcal I_{\mathrm{ev}}}(d))^2}
 {c^2}
 \right)\le2
 \right\}
\end{equation}
by bisection, and the centered exceedance frequency
\begin{equation}
\widehat p_t(d_\star)
 =\frac{1}{|\mathcal I_{\mathrm{ev}}|}
 \sum_{r\in\mathcal I_{\mathrm{ev}}}
 \mathbf 1\left\{
 |X_r(d_\star)
 -\overline X_{\mathcal I_{\mathrm{ev}}}(d_\star)|>t
 \right\}.
\end{equation}
Directional medians and $10$--$90\%$ bands summarize
$\{\widehat\sigma_{\mathcal I_{\mathrm{ev}}}(d):
d\in\mathcal D\}$. Evaluation-set means are used separately
to measure bias.

\textbf{Interpretation and uncertainty.} 
The statistics we compute in our experiments are finite-sample diagnostics over a
finite direction set, and hence they cannot exactly certify population
subgaussian parameters or maximize over the entire unit
sphere. Similarly, the moment diagnostic uses  orders $2,4,6$
and is Gaussian-calibrated, so its population value for
$N(0,\sigma^2)$ is $\sigma$. The Orlicz norm uses a
different normalization. For $G\sim N(0,\sigma^2)$,
$\|G\|_{\psi_2}=\sqrt{8/3}\,\sigma$, whereas a symmetric
Rademacher variable has Orlicz norm $1/\sqrt{\log 2}$. The error bars in the main sweeps show the standard deviation
of five batchwise maxima over the tested directions.
The batches use evaluation samples centered by their
full evaluation-set means.  The standardized-tail experiment in Figure~\ref{fig:shape} uses a
separate descriptive protocol. It computes uncentered
moment diagnostics on all $R$ walks, selects the
maximizing tested direction, and rescales that
direction's empirical survival curve using its
diagnostic from the same samples. Figure~\ref{fig:shape} therefore
describes the observed normalized tail shape.

\textbf{Discussion.} Figure~\ref{fig:bits} shows that the centered directional diagnostics
decrease as precision increases and then approach a plateau;
its directional bands describe variation across tested
directions. Figure~\ref{fig:shape} compares empirically rescaled,
uncentered survival curves across problem sizes and
suggests similar behavior over part of the observed range,
without establishing a common population tail law.   Figure~\ref{fig:noise} compares rounding with independently injected
noise and explores candidate rescalings.   Finally, Figure~\ref{fig:steplanding} uses $150$ walks with $m=2$, $n=200$, and five retained fraction bits. Along each float64 reference trajectory, we recompute directions in low precision at identical states and summarize their coefficient differences by the number of free nonpivot coordinates. The median and interquartile ranges reveal a pronounced error spike near $k=m$.

\textbf{Input representation.}
Let $B_0$ be the original input and $\widehat B$
its working-precision representation. As the
experiments start at $z_0=0$, we have $B_0z=\widehat Bz+(B_0-\widehat B)z$. So, discrepancy measured against $B_0$ includes
input-representation error as well as discrepancy
relative to the represented matrix.

\textbf{Combined bias and fluctuation diagnostic.}
For the scaling-collapse plots, define
\[
\widehat b_{\max}
=\max_{d\in\mathcal D}
 |\overline X_{\mathcal I_{\mathrm{ev}}}(d)|,
\qquad
\widehat S
=\left(
 \widehat\sigma_{\mathcal I_{\mathrm{ev}}}(d_\star)^2
 +\widehat b_{\max}^{\,2}
\right)^{1/2}.
\]
These terms may correspond to different directions.
So, $\widehat S$ combines centered fluctuations and
bias, rather than estimating the Orlicz norm or
uncentered moment diagnostic of a single projection, such that
$\max\{
 \widehat\sigma_{\mathcal I_{\mathrm{ev}}}(d_\star),
 \widehat b_{\max}
\}
\le\widehat S
\le\sqrt2\,
\max\{
 \widehat\sigma_{\mathcal I_{\mathrm{ev}}}(d_\star),
 \widehat b_{\max}
\}$.

\begin{figure}[hbt!]
  \centering
  \includegraphics[width=0.9\linewidth]{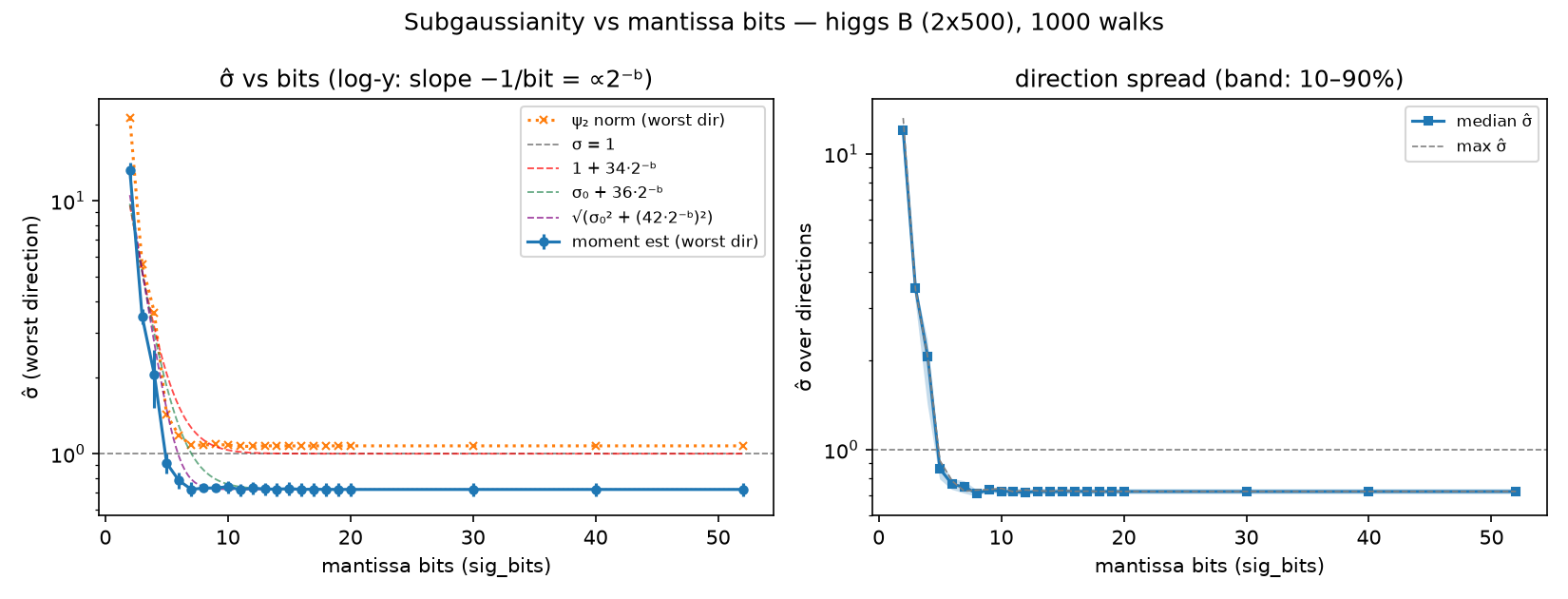}
  \caption{\textbf{Worst-direction $\hat{\sigma}$ vs.\ mantissa bits at fixed
  $n = 500$.} The plateau ends near $b \approx 8$. Three one-parameter fits
  describe the low-precision side nearly equally well; of these, the
  quadrature form $\sqrt{\sigma_{0}^{2} + (42 \cdot 2^{-b})^{2}}$ is the one
  consistent with the $n$-sweep. The dotted $\psi_{2}$ norm sits about
  $1.4\times$ above the moment estimate throughout. Right: the median and 10--90\%
  direction band.}
  \label{fig:bits}
\end{figure}

\begin{figure}[hbt!]
  \centering
  \includegraphics[width=0.8\linewidth]{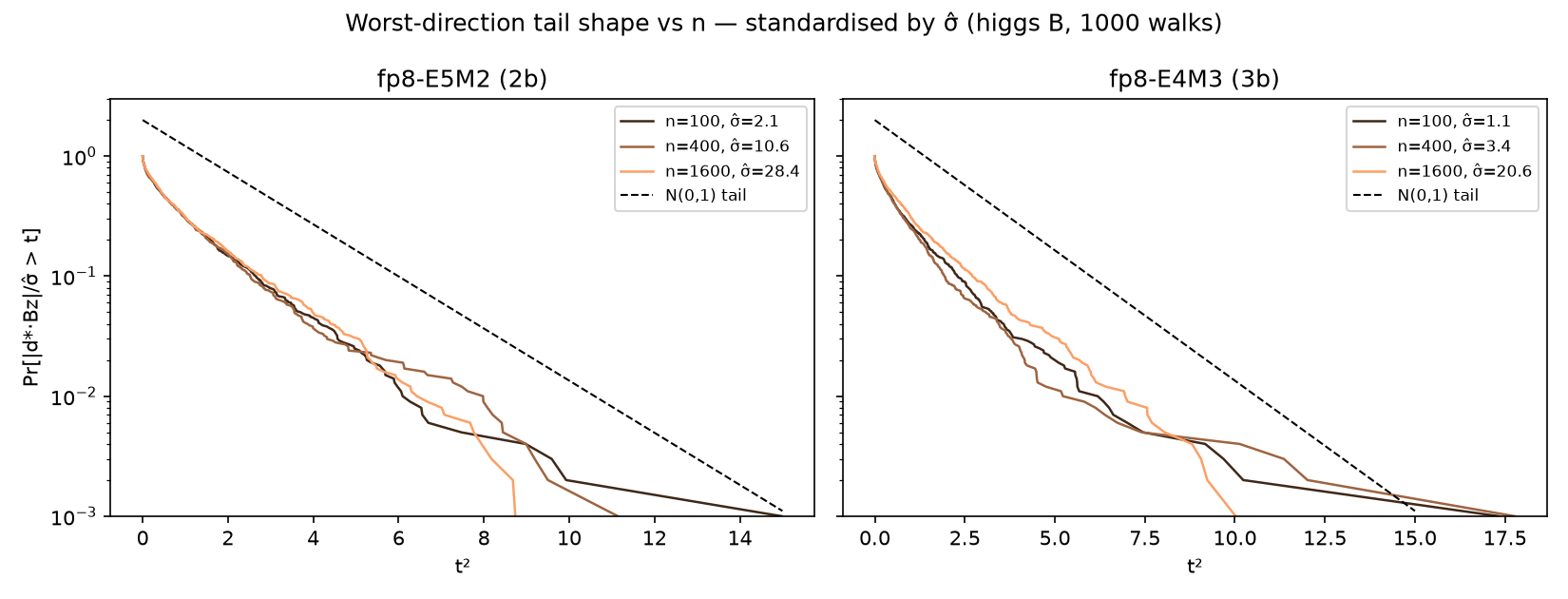}
  \caption{\textbf{Standardized empirical tails at two and three retained
fraction bits for $n=100,400,1600$.} We plot
$\widehat{\Pr}\{|d_*^\top Bz|/\widehat\sigma>t\}$ using the
direction and scale selected from the same samples. The dashed
curve is $2e^{-t^2/2}$, which we use as a reference envelope. The curves show broadly similar
behavior over part of the observed range, with deviations in the
empirical far tail.}
  \label{fig:shape}
\end{figure}

\begin{figure}[hbt!]
  \centering
  \includegraphics[width=0.9\linewidth]{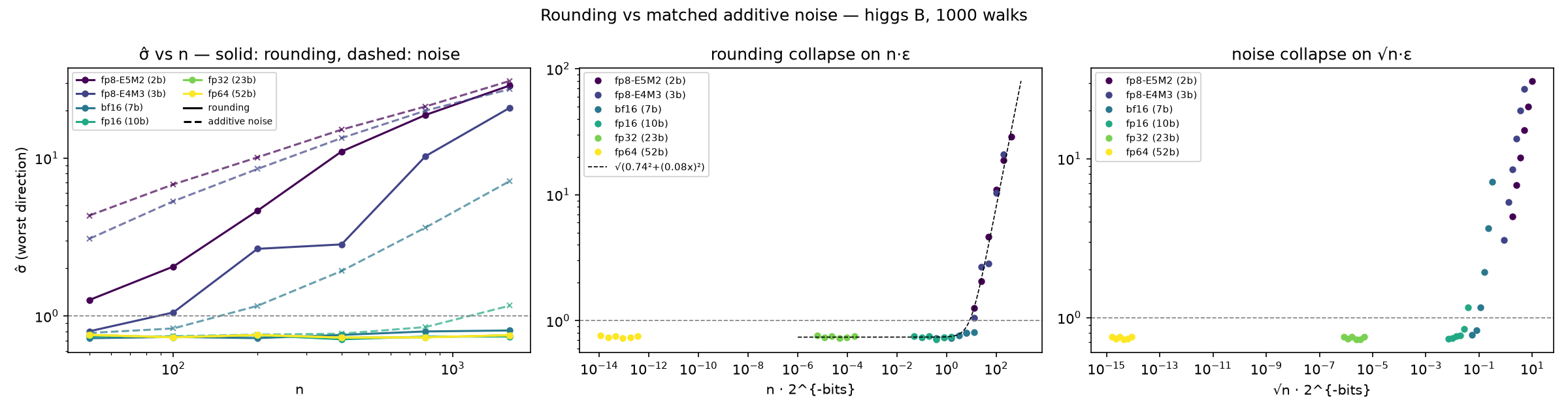}
  \caption{\textbf{Rounding is not additive noise.} Solid curves run the
  walk's arithmetic at $b$ mantissa bits, dashed curves inject i.i.d.\ noise
  of the same magnitude $2^{-b}$ per step. Rounding is reused coherently and
  grows like $n\varepsilon$, so the cells collapse on the $n\varepsilon$ axis
  (middle). Independent noise accumulates like $\sqrt{n}\,\varepsilon$ and
  does not collapse (right).}
  \label{fig:noise}
\end{figure}

\begin{figure}[hbt!]
  \centering
  \includegraphics[width=0.8\linewidth]{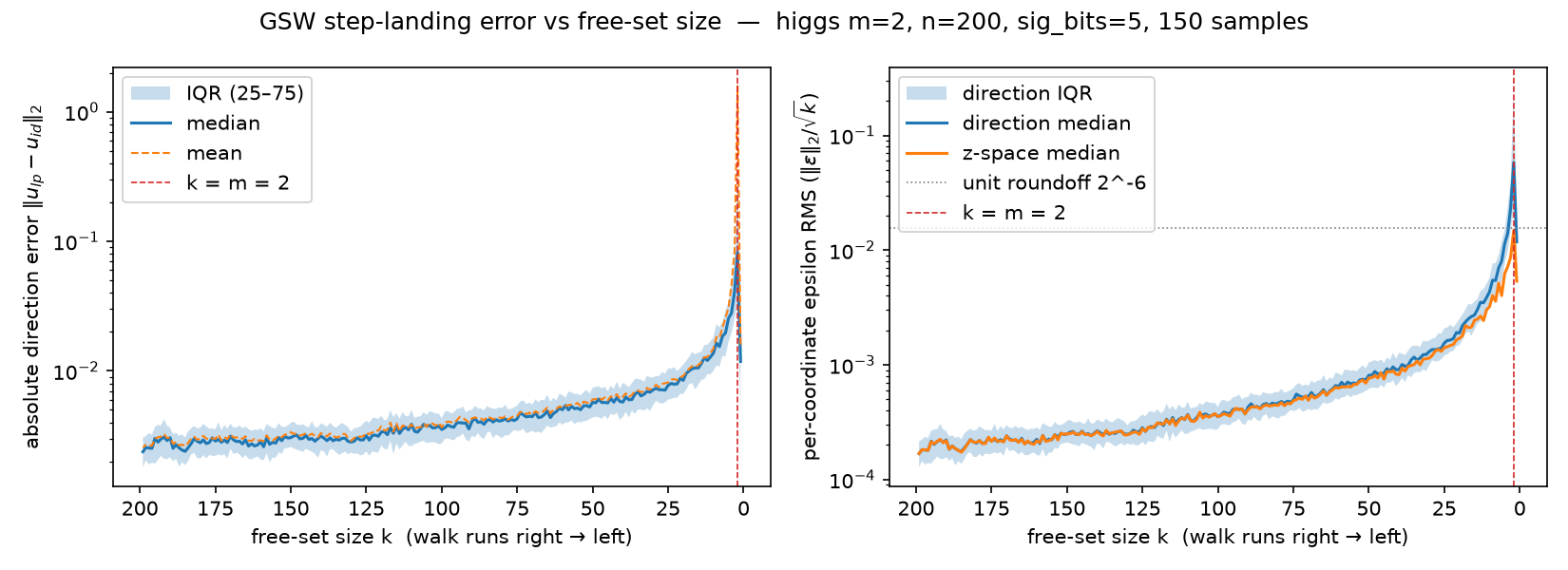}
  \caption{Coefficient differences on the float64 reference
  trajectory (150 walks, $m = 2$, $n = 200$, 5 mantissa bits). The error
  stays a small fraction of the unit roundoff for most of the walk, then
  spikes by orders of magnitude when the free set reaches $k \approx m$ and
  $B_{A \setminus p}$ becomes nearly singular. We show the median and IQR
  because the rare blowups dominate the mean. This diagnostic compares coefficient vectors, not the image-error quantity $\|B(\widehat u_t-u_t)\|_2$, and it is
evaluated on reference rather than perturbed trajectories.}
  \label{fig:steplanding}
\end{figure}

\begin{figure}
    \centering
    \includegraphics[width=0.7\linewidth]{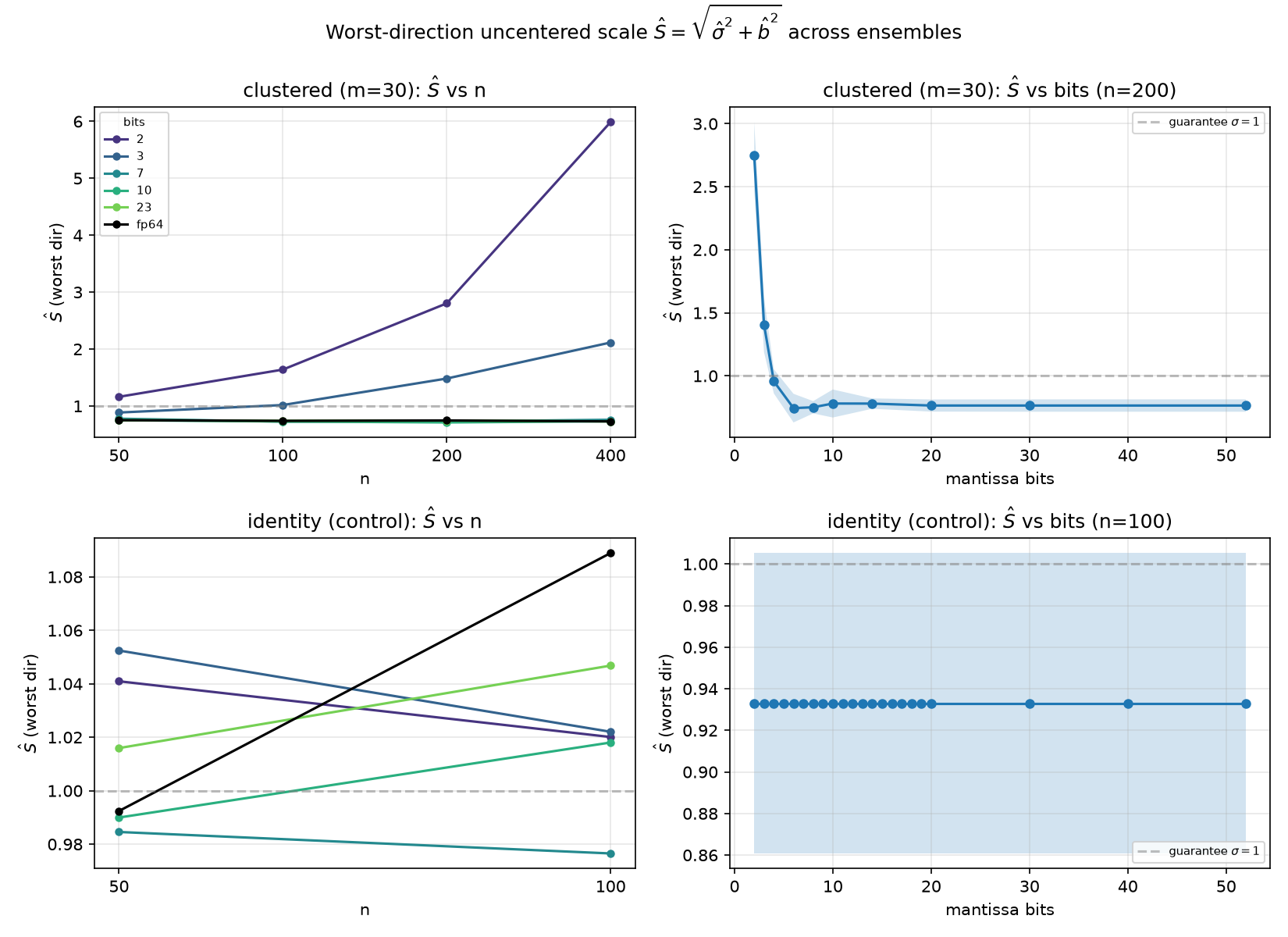}
    \caption{\textbf{Worst-direction uncentered scale $\hat{S} =\sqrt{\hat{\sigma}^2 + \hat{b}^2_{max}}$ (1000 walks per cell, 64 directions)}. \textbf{Top:} clustered matrix $B\in\mathbb R^{30\times n}$: (\emph{left}), $\hat{S}$ vs $n$ at $6$ precisions (fp64 black) (\emph{right}) $\hat{S}$ vs mantissa bits at $n=200$, and the shaded band shows
$\pm 1.96$ standard deviations of five batchwise direction-maxima of
$\hat{\sigma}$. The precision knee at $b\approx 5$ matches the $n\epsilon$ scaling observed on Higgs. \textbf{Bottom:}  identity matrix $B=I_n$ is a degenerate control: $e_p$ is orthogonal to the active block, so the least-squares step is exact at every precision and $\hat S$ is flat across $n$ and bits.}
    \label{fig:placeholder}
\end{figure}

\textbf{Adversarial stress test.}
Figure~\ref{fig:collapse} is a numerical stress test on
Theorem~\ref{thm:bounded-error-lower-bound}: we inject
$a = 2^{-b}$ on active non-pivot coordinates for the first
$\lfloor n/4 \rfloor$ steps \emph{without} the cap $a < 1/(8n)$ (violated cells
are drawn as open markers), and directions, probabilities, and updates are rounded at $b$ bits rather than exact. In our simulations, runs may also terminate within $n$ steps with unfrozen coordinates, on which we still report $B_0 z$. It thus measures
adversarial bias combined with rounding error.

\subsection{Dtype-based solver comparisons}
\label{subsection: dtype comparisons}
In this subsection, a walk is classified as failed if it encounters a nonfinite
direction or branch probability, a Cholesky or rank-one downdate
breakdown, or exhausts the step budget before reaching a complete
signing. 
The experiments we provide here use dtype-based float64, float32, and float16
arithmetic, including pairwise reductions in the selected dtype.
The test families comprise Gaussian matrices, matrices with
controlled condition numbers, near-collinear column pairs, and
Harshaw designs, with $n\in\{50,200,500\}$. Attempted sample counts
vary by solver and configuration. The repository reports
6.9 hours of wall-clock time with 32 workers and 211 CPU-hours
of walks for the EPYC grid; this compute budget is separate
from that of the main precision sweeps.\looseness=-1
 
Here, the failure rates include all attempted walks. Orlicz and bias
diagnostics exclude failed walks and therefore describe outputs
conditional on successful execution, and discrepancy is evaluated using the original
float64 matrix. 
Table~\ref{tab:solver-refactoring} reproduces
selected configurations from the stored EPYC-grid summaries, and
the complete configuration tables are provided in the repository.
\begin{table}[t]
\centering
\footnotesize
\setlength{\tabcolsep}{4pt}
\begin{tabular}{rlrrrrr}
\hline
$n$ & Solver & Attempts & Failures & Fail.\ (\%)
    & Successes & $\widehat K_{\mathrm{succ}}$ \\
\hline
50  & QR             & 2000 &    0 &  0.00 & 2000 & 0.756 \\
50  & Cholesky       & 2000 & 1948 & 97.40 &   52 & 1.24  \\
50  & Cholesky, $k=10$ & 2000 &  0 &  0.00 & 2000 & 1.37  \\
50  & Explicit inverse & 2000 & 0 &  0.00 & 2000 & 0.754 \\
\hline
200 & QR             & 1000 &    0 &   0.00 & 1000 & 0.653 \\
200 & Cholesky       & 2000 & 2000 & 100.00 &    0 & ---   \\
200 & Cholesky, $k=10$ & 2000 &  1 &   0.05 & 1999 & 6.72  \\
200 & Explicit inverse & 2000 & 0 &   0.00 & 2000 & 1.19  \\
\hline
500 & QR             &   20 &    0 &   0.00 &   20 & 0.582 \\
500 & Cholesky       & 2000 & 2000 & 100.00 &    0 & ---   \\
500 & Cholesky, $k=10$ & 2000 & 35 &   1.75 & 1965 & 19.2  \\
500 & Explicit inverse & 300 &  0 &   0.00 &  300 & 3.36  \\
\hline
\end{tabular}
\caption{
Float16 solver comparisons on correlated Harshaw designs with
$d=5$ and $\phi=0.01$.
QR recomputes the least-squares direction from scratch;
$k=10$ denotes Cholesky refactoring every ten steps.
$\widehat K_{\mathrm{succ}}$ is the maximum empirical
Orlicz estimate over 20 fixed random unit directions and the
top and bottom left singular-vector directions, using only
successful walks. A dash indicates no successful walks. The QR estimate at $n=500$ uses  20 walks.
}
\label{tab:solver-refactoring}
\end{table} For these configurations, refactoring substantially reduces
the observed failure rate but does not eliminate discrepancy
inflation among successful walks. At $n=500$, its 35 failures
comprise 26 Cholesky breakdowns and 9 downdate breakdowns.
Across the 36 float16 Harshaw configurations with refactoring,
the largest observed failure rate is 1.75\%, and all remaining
configurations have rates at most 0.05\%.\looseness=-1

\end{document}